\documentclass[10pt]{amsart}
\pdfoutput=1
\usepackage{amsmath,amssymb,amsfonts,bussproofs}
\usepackage{hyperref}
\usepackage[all]{hypcap} 
\usepackage{cmap}
\usepackage[cp1251]{inputenc} 
\usepackage[russian,english]{babel}
\usepackage[mathscr]{eucal}

\newtheorem{lemma}{Lemma}
\newtheorem{theorem}{Theorem}
\newtheorem{proposition}{Proposition}

\newcommand{\LukA}{\textnormal{\L}\forall}
\newcommand{\RPLA}{\textnormal{RPL}\forall}
\newcommand{\GLukA}{\textnormal{G\L}\forall}
\newcommand{\GoneLukA}{\textnormal{G}^1\textnormal{\L}\forall}
\newcommand{\GtwoLukA}{\textnormal{G}^2\textnormal{\L}\forall}
\newcommand{\GthreeLukA}{\textnormal{G}^3\textnormal{\L}\forall}
\newcommand{\TthreeLukA}{\textnormal{T}^3\textnormal{\L}\forall} 
\newcommand{\RePl}[3]{[#1]^{#2}_{#3}} 
\newcommand{\SpV}[1]{\mathfrak{#1}}

\begin{document}
\thispagestyle{empty}

\title[A repetition-free hypersequent calculus for RPL$\forall$]
{A repetition-free hypersequent calculus for first-order rational Pavelka logic}
\author[A. S. Gerasimov]{Alexander S. Gerasimov}
\email{\href{mailto:alexander.s.gerasimov@ya.ru}
       {\texttt{alexander.s.gerasimov@ya.ru}}}

\maketitle

{\small
\begin{quote}
\noindent{\bf Abstract. } 
We present a hypersequent calculus $\text{G}^3\text{\L}\forall$
for first-order infinite-valued {\L}ukasiewicz logic and
for an extension of it, first-order rational Pavelka logic; 
the calculus is intended for bottom-up proof search.
In $\text{G}^3\text{\L}\forall$,
there are no structural rules, all the rules are invertible,
and designations of multisets of formulas are not repeated 
in any premise of the rules.
The calculus $\text{G}^3\text{\L}\forall$ proves 
any sentence that is provable in at least one of 
the previously known hypersequent calculi for the given logics.
We study proof-theoretic properties of $\text{G}^3\text{\L}\forall$
and thereby provide foundations for proof search algorithms.
\medskip


\end{quote}
}

\section{Introduction}

First-order infinite-valued {\L}ukasiewicz logic $\LukA$ and 
an extension of it by rational truth constants,
first-order rational Pavelka logic $\RPLA$, 
are among the fundamental fuzzy logics \cite{Hajek1998, HMFL2011, HMFL2015}
and are considered in the given paper from the standpoint of proof search.

Hilbert-type calculi for the logics under consideration are widely used
(see, e.g., \cite{Hajek1998, HMFL2011}), 
but such calculi are unfit for bottom-up proof search.
For $\LukA$, we also know the hypersequent calculus $\GLukA$ 
\cite{BaazMetcalfe2010, MOG2009} 
with structural rules, which make it unsuitable for bottom-up proof search.

On the basis of the calculus $\GLukA$ from \cite{BaazMetcalfe2010, MOG2009} 
and tableau calculi from \cite{Ger2016}, in \cite{Ger2017}
we introduced hypersequent calculi $\GoneLukA$ and $\GtwoLukA$ 
for the logic $\RPLA$ and hence for $\LukA$. 
The calculi $\GoneLukA$ and $\GtwoLukA$ do not have structural rules;
the latter is a noncumulative variant of the former, which is cumulative,
i.e., preserves the conclusion of each inference rule in its premises.
Any $\GLukA$-provable sentence is provable in $\GoneLukA$;
and any prenex $\LukA$-sentence is provable or unprovable in 
$\GLukA$, $\GoneLukA$, and $\GtwoLukA$ simultaneously.
Also in \cite{Ger2017}, a family of proof search algorihms is described;
given a prenex $\GtwoLukA$-provable sentence, such an algorithm constructs 
some proof for it in a tableau modification of the calculus $\GtwoLukA$.

A defect of $\GtwoLukA$  
(which does not appear in proving prenex sentences)  
is that designations of multisets of formulas are repeated
in each premise of two quantifier rules.
The defect causes repeating some work during bottom-up proof search
and prevented us from establishing desirable proof-theoretic properties 
for the calculus $\GtwoLukA$, in particular, invertibility 
of one of its rules.

In the present paper, we introduce a noncumulative hypersequent calculus
$\GthreeLukA$ for the logic $\RPLA$.
There are no structural rules in the calculus; and
designations of multisets of formulas are not repeated 
in any premise of its rules.
The last feature of the calculus allows us to call it and each of its rules
repetition-free.

This paper is organized as follows.
In the rest of this section, we define the syntax and semantics
of the logics $\LukA$ and $\RPLA$, as well as some notation.
In Section \ref{Sec:GthreeSound}, 
we formulate the calculus $\GthreeLukA$ and prove its soundness.
In Section \ref{Sec:Invert}, 
we establish the invertibility of all the rules of $\GthreeLukA$
and show that any $\GoneLukA$-provable sentence is provable in $\GthreeLukA$.
In Section \ref{Sec:Transform}, 
we investigate transformations of $\GthreeLukA$-proofs
according to proof search tactics and 
thereby provide foundations for various proof search algorithms.
In Section \ref{Sec:MidHS}, 
we prove the mid-hypersequent theorem for $\GthreeLukA$;
show that any prenex $\RPLA$-sentence is provable or unprovable in 
$\GoneLukA$, $\GtwoLukA$, and $\GthreeLukA$ simultaneously;
and establish undecidability of $\GthreeLukA$. 

Let us describe the syntax and semantics of the logics under consideration.
We fix an arbitrary signature, which may contain predicate and function
symbols of any nonnegative arities.

\emph{Terms} are defined in the standard manner.
\emph{Atomic} $\LukA$- and $\RPLA$-\emph{formulas} are predicate symbols
with argument terms, as well as truth constants: 
in $\LukA$, the only truth constant $\bar{0}$;
and in $\RPLA$, truth constants $\bar{r}$ for all rational numbers $r \in [0,1]$ 
(where $[0,1]$ is an interval of real numbers).
$\LukA$- and $\RPLA$-\emph{formulas} are built as usual from
atomic $\LukA$- and $\RPLA$-formulas, respectively, 
using the \emph{logical symbols}: 
the binary connective $\to$ and the quantifiers $\forall, \exists$.

The notion of an \emph{interpretation} ${\langle \mathcal{D}, \mu \rangle}$
differs from the classical notion of the same name only in that
the map $\mu$ takes each $n$-ary predicate symbol $P$ to a predicate
${\mu(P): \mathcal{D}^n \to [0,1]}$.
Given an interpretation ${\langle \mathcal{D}, \mu \rangle}$,
a \emph{valuation} is a map of the set of all (individual) variables to
the domain $\mathcal{D}$ of the interpretation.
For a valuation $\nu$, a variable $x$, and ${d\in\mathcal{D}}$, by
$\nu[x\mapsto d]$ we denote the valuation that may differ from $\nu$ only on $x$
and meets the condition \,${\nu[x\mapsto d](x) = d}$.

The \emph{value $|t|_{M,\nu}$ of a term $t$} under an interpretation $M$ and 
a valuation $\nu$ is defined as usual.
The \emph{truth value $|C|_{M,\nu}$ of an $\RPLA$-formula $C$} 
under an interpretation ${M = \langle \mathcal{D}, \mu \rangle}$ and 
a valuation $\nu$ is defined as follows:

(1)~${|\bar{r}|_{M,\nu}=r}$;

(2)~${|P(t_1,\dots,t_n)|_{M,\nu} = \mu(P)(|t_1|_{M,\nu},\dots,|t_n|_{M,\nu})}$
for an $n$-ary predicate symbol $P$ and terms $t_1,\dots,t_n$;

(3)~${|A \to B|_{M,\nu} = \min(1-|A|_{M,\nu}+|B|_{M,\nu}, \,1)}$;

(4)~${|\forall x A|_{M,\nu} = \inf_{d\in\mathcal{D}} |A|_{M,\nu[x\mapsto d]}}$;

(5)~${|\exists x A|_{M,\nu} = \sup_{d\in\mathcal{D}} |A|_{M,\nu[x\mapsto d]}}$.

An $\RPLA$-formula $C$ is called \emph{valid} 
(also written ${\vDash C}$) if ${|C|_{M,\nu} = 1}$ 
for every interpretation $M$ and every valuation~$\nu$.

Note that the logic $\RPLA$ allows us to express partial truth of statements
in the following way \cite[Section 3.3]{Hajek1998}. 
Given a rational number ${r \in [0,1]}$ and an $\RPLA$-formula $A$, we have: 
\:(a)~for a fixed interpretation $M$ and a fixed valuation~$\nu$:
\,${r \leqslant |A|_{M,\nu}}$\, iff \,${|\bar{r} \to A|_{M,\nu} = 1}$;
\:(b)~${r \leqslant |A|_{M,\nu}}$ for every interpretation $M$ and every
valuation~$\nu$\, iff \,${\vDash (\bar{r} \to A)}$.

The result of substituting a term $t$ for all free occurrences of 
a variable $x$ in an $\RPLA$-formula $A$ is denoted by ${\RePl{A}{x}{t}}$.
By a proof in a calculus considered below, we understand a proof tree.
The provability of an object $\alpha$ in a calculus $\mathfrak{C}$ is denoted
by \,${\vdash_\mathfrak{C} \alpha}$.

The calculi $\GLukA$, $\GoneLukA$, and $\GtwoLukA$ are formulated in 
Sections 2.1, 2.2, and 3.1 of the paper \cite{Ger2017}.

\section{The repetition-free calculus $\GthreeLukA$ and its soundness} 
\label{Sec:GthreeSound}

Basically, we obtain $\GthreeLukA$ from the calculus $\GtwoLukA$, 
defined in \cite[Section 3.1]{Ger2017}, 
by replacing its rules ${(\forall\Rightarrow)^2}$ and 
${(\Rightarrow\exists)^2}$ with repetition-free ones. 

We will work with a fixed signature that includes 
a countable set of nullary function symbols called \emph{parameters}.

Semipropositional variables defined in \cite[Section 2.2]{Ger2017}
are now called \emph{semipropositional variables of type}~1 and
are denoted by $\SpV{p}, \SpV{p}_0, \SpV{p}_1,\ldots$~.
In addition to them, we introduce a countable set of new words
called \emph{semipropositional variables of type}~0 and
denoted by $\SpV{q}, \SpV{q}_0, \SpV{q}_1, \ldots$~.

The way of obtaining the new repetition-free rules and the role of
semipropositional variables used in them will be revealed in the proofs of 
Lemmas \ref{GLem:CorrSemInvGthreeRules} and \ref{GLem:SemGtwoGthreeRules} below.

We define an \emph{hs-interpretation} as an interpretation
$\langle \mathcal{D}, \mu \rangle$ in which the map $\mu$ additionally
takes each semipropositional variable of type~0
to a real number from $[0, +\infty)$ and 
each semipropositional variable of type~1 to a real number from $(-\infty, 1]$.

Taking into account that by semipropositional variables
we now mean semipropositional variables of both types,
the following definitions and abbreviations given in \cite[Section 2.2]{Ger2017}
preserve their forms:
the definitions of an \emph{atom}, a \emph{formula}, a \emph{sequent}, 
a \emph{hypersequent}, a \emph{member} of a sequent, an \emph{atomic} sequent;
the abbreviations $|\Gamma|_{M,\nu}$ and ${|\Gamma\Rightarrow\Delta|_{M,\nu}}$
(for a finite multiset $\Gamma$ of formulas, a sequent
${\Gamma\Rightarrow\Delta}$, an hs-interpretation $M$, and a valuation $\nu$); 
the definitions of a \emph{true} sequent 
(under an hs-interpretation and a valuation), 
a \emph{valid} hypersequent (with the abbreviation \,${\vDash \mathcal{H}}$ 
for such a hypersequent $\mathcal{H}$),
\emph{sound} and \emph{semantically invertible} rules.

In the sequel, let the letters $A$, $B$, and $C$ denote any $\RPLA$-formulas, 
$F$ a formula,
$\Gamma$, $\Delta$, $\Pi$, and $\Sigma$ any finite multisets of formulas, 
$S$ a sequent,
$\mathcal{G}$ and $\mathcal{H}$ any hypersequents, 
$t$ a closed term,
$a$ a parameter;  
all these letters may have subscripts.

The inference rules of the calculus $\GthreeLukA$ are:
\medskip
\begin{center}
$\dfrac{\mathcal{G}\,|\, \Gamma, \SpV{p} \Rightarrow \Delta\,|\,B \Rightarrow \SpV{p}, A}
 {\mathcal{G}\,|\,\Gamma, A\to B \Rightarrow \Delta}~(\to\Rightarrow)^3$, \qquad
$\dfrac{\mathcal{G}\,|\, \Gamma \Rightarrow \Delta;  \quad  \mathcal{G}\,|\, \Gamma,A \Rightarrow B,\Delta}
 {\mathcal{G}\,|\,\Gamma \Rightarrow A\to B, \Delta}~(\Rightarrow\to)^3$,\\[\medskipamount]
$\dfrac{\mathcal{G}\,|\, \Gamma, \SpV{p} \Rightarrow\Delta\,|\,\forall x A \Rightarrow \SpV{p}\,|\,\RePl{A}{x}{t} \Rightarrow \SpV{p}}
 {\mathcal{G}\,|\, \Gamma, \forall x A \Rightarrow\Delta}~(\forall\Rightarrow)^3$, \qquad
$\dfrac{\mathcal{G}\,|\, \Gamma \Rightarrow \RePl{A}{x}{a}, \Delta}
 {\mathcal{G}\,|\,\Gamma \Rightarrow \forall x A, \Delta}~(\Rightarrow\forall)^3$,\\[\medskipamount]
$\dfrac{\mathcal{G}\,|\, \Gamma \Rightarrow \SpV{q}, \Delta\,|\,\SpV{q} \Rightarrow \exists x A\,|\,\SpV{q} \Rightarrow \RePl{A}{x}{t}}
 {\mathcal{G}\,|\, \Gamma \Rightarrow \exists x A, \Delta}~(\Rightarrow\exists)^3$, \qquad
$\dfrac{\mathcal{G}\,|\, \Gamma, \RePl{A}{x}{a} \Rightarrow\Delta}
 {\mathcal{G}\,|\,\Gamma, \exists x A \Rightarrow\Delta}~(\exists\Rightarrow)^3$,
\end{center}
\medskip
where $\SpV{p}$ (resp. $\SpV{q}$) does not occur in the conclusion of 
${(\to\Rightarrow)^3}$ or ${(\forall\Rightarrow)^3}$
(resp. ${(\Rightarrow\exists)^3}$) and is called 
the \emph{proper semipropositional variable} 
of an application of the corresponding rule;
$t$ is called the \emph{proper term} of an application of 
${(\forall\Rightarrow)^3}$ or ${(\Rightarrow\exists)^3}$;
$a$ does not occur in the conclusion of ${(\Rightarrow\forall)^3}$ or
${(\exists\Rightarrow)^3}$ and is called the \emph{proper parameter}
of an application of the corresponding rule.

An \emph{axiom} of the calculus $\GthreeLukA$ is an arbitrary hypersequent 
in which, for any hs-interpretation $M$ and any valuation $\nu$,
there exists an atomic sequent that is true under $M$ and~$\nu$.
Note that axioms of $\GthreeLukA$ can be recognized in much the same way as
described in \cite[Section 4.2]{Ger2017}. 

A \emph{$\GthreeLukA$-proof of} (\emph{for}) \emph{an $\RPLA$-formula $A$}
is a $\GthreeLukA$-proof of the hypersequent \,${\Rightarrow A}$.

The following definitions and notation 
given at the end of \cite[Section 2.2]{Ger2017} 
carry over to the calculus $\GthreeLukA$:
the definitions of a \emph{backward application} 
(or a \emph{counter-application}) of a rule,
a \emph{principal} formula (sequent) occurrence,
and an \emph{ancestor} of a formula (sequent) occurrence;
and the convention for designating a proof of a hypersequent over
an occurrence of it in a proof tree.

Suppose $D$ is a $\GthreeLukA$-proof, and $\mathcal{G}$ is a hypersequent.
To get a $\GthreeLukA$-proof $D'$, in $D$, we rename
all proper semipropositional variables occurring in $\mathcal{G}$
and all proper parameters occurring in $\mathcal{G}$ to new distinct ones. 
Then by $D\,|\,\mathcal{G}$ we denote the $\GthreeLukA$-proof
obtained from $D'$ by appending ``$|\,\mathcal{G}$'' to 
each node hypersequent of $D'$.
(For our use of such an abbreviation, 
it does not matter how we perform renaming above.)

\begin{lemma}  \label{GLem:CorrSemInvGthreeRules}
Each inference rules of the calculus $\GthreeLukA$ is
sound and semantically invertible.
\end{lemma}

\begin{proof}
From assertions (1)--(4) of Lemma \ref{GLem:SemGtwoGthreeRules} stated below,
it follows that the rules ${(\to\Rightarrow)^3}$, ${(\Rightarrow\to)^3}$, 
${(\Rightarrow\forall)^3}$, and ${(\exists\Rightarrow)^3}$ 
are sound and semantically invertible.

Any application of the rule ${(\forall\Rightarrow)^3}$ can be represented
as two applications of the rules 
$$\dfrac{\mathcal{G}_0 \,|\, \Gamma_0, \forall x A \Rightarrow \Delta_0 \,|\, \Gamma_0, \RePl{A}{x}{t} \Rightarrow \Delta_0}
  {\mathcal{G}_0 \,|\, \Gamma_0, \forall x A \Rightarrow \Delta_0}~(\forall\Rightarrow)^2_0
  \quad \text{and} \quad
  \dfrac{\mathcal{G} \,|\, \Gamma, \SpV{p} \Rightarrow \Delta \,|\, B \Rightarrow \SpV{p}}
  {\mathcal{G} \,|\, \Gamma, B \Rightarrow \Delta}~(\text{den}_1),
$$
where $\SpV{p}$ does not occur in the conclusion of the last rule,
as follows:\footnote{
  We obtained the repetition-free rule ${(\forall\Rightarrow)^3}$ in this way.
  The rule ${(\forall\Rightarrow)^2_0}$ differs from the rule
  ${(\forall\Rightarrow)^2}$ of the calculus $\GtwoLukA$ only in that 
  semipropositional variables of type~0 may occur in 
  a premise and conclusion of ${(\forall\Rightarrow)^2_0}$.
  The rule $(\text{den}_1)$ is a nonstandard variant of the density rule,
  cf. \cite[Section 4.5]{MOG2009}. 
}
\begin{center}
\def\ScoreOverhang{0pt}
\AxiomC{$\mathcal{G}\,|\, \Gamma, \SpV{p} \Rightarrow\Delta\,|\,\forall x A \Rightarrow \SpV{p}\,|\,\RePl{A}{x}{t} \Rightarrow \SpV{p}$}
\RightLabel{$(\forall\Rightarrow)^2_0$}
\UnaryInfC{$\mathcal{G}\,|\,\Gamma, \SpV{p} \Rightarrow\Delta \,|\, \forall x A \Rightarrow \SpV{p}$}
\RightLabel{$(\text{den}_1)$.}
\UnaryInfC{$\mathcal{G}\,|\,\Gamma, \forall x A \Rightarrow\Delta$}
\DisplayProof
\end{center}
By assertion (5) of Lemma \ref{GLem:SemGtwoGthreeRules},
the rule ${(\forall\Rightarrow)^2_0}$ is sound;
and it is semantically invertible, since its premise includes its conclusion.
By assertion $(\widetilde{5})$ of Lemma~\ref{GLem:SemGtwoGthreeRules},
the rule $(\text{den}_1)$ is sound and semantically invertible.
So ${(\forall\Rightarrow)^3}$ is sound and semantically invertible.

Any application of the rule ${(\Rightarrow\exists)^3}$ can be represented
as two applications of the rules 
$$\dfrac{\mathcal{G}_0 \,|\, \Gamma_0 \Rightarrow \exists x A, \Delta_0 \,|\, \Gamma_0 \Rightarrow \RePl{A}{x}{t}, \Delta_0}
  {\mathcal{G}_0 \,|\, \Gamma_0 \Rightarrow \exists x A, \Delta_0}~(\Rightarrow\exists)^2_0
  \quad \text{and} \quad
  \dfrac{\mathcal{G} \,|\, \Gamma \Rightarrow \SpV{q}, \Delta \,|\, \SpV{q} \Rightarrow B}
  {\mathcal{G} \,|\, \Gamma \Rightarrow B, \Delta}~(\text{den}_0),
$$
where $\SpV{q}$ does not occur in the conclusion of the last rule, thus:
\begin{center}
\def\ScoreOverhang{0pt}
\AxiomC{$\mathcal{G}\,|\, \Gamma \Rightarrow \SpV{q}, \Delta \,|\, \SpV{q} \Rightarrow \exists x A \,|\, \SpV{q} \Rightarrow \RePl{A}{x}{t}$}
\RightLabel{$(\Rightarrow\exists)^2_0$}
\UnaryInfC{$\mathcal{G} \,|\, \Gamma \Rightarrow \SpV{q}, \Delta \,|\, \SpV{q} \Rightarrow \exists x A$}
\RightLabel{$(\text{den}_0)$.}
\UnaryInfC{$\mathcal{G} \,|\, \Gamma \Rightarrow \exists x A, \Delta$}
\DisplayProof
\end{center}
Then from assertions $(6)$ and $(\widetilde{6})$ 
of Lemma \ref{GLem:SemGtwoGthreeRules}, it follows that 
${(\Rightarrow\exists)^3}$ is sound and semantically invertible.
\end{proof}

For an hs-interpretation $M$, 
a semipropositional variable $\SpV{r}$ of type~0 (resp. type~1), 
and a real number ${r \in [0, +\infty)}$ (resp. ${r \in (-\infty, 1]}$),
by ${M[\SpV{r}\mapsto r]}$ we denote the hs-inter\-pre\-ta\-tion that
interprets $\SpV{r}$ by $r$ and does not differ from $M$ in any other respect.

\begin{lemma}  \label{GLem:SemGtwoGthreeRules}
Let \,$\Gamma$ and $\Delta$ be finite multisets of formulas;
$A$ and $B$ be $\RPLA$-for\-mulas;
$y$ be a variable not occurring in $\Gamma$, $\Delta$, $A$;
$\SpV{p}$ and $\SpV{q}$ be semipropositional variables 
\textnormal{(}of type~1 and type~0, respectively\textnormal{)} 
not occurring in $\Gamma$, $\Delta$, $A$, $B$;
$M$ be an hs-interpretaton with domain $\mathcal{D}$;
and $\nu$ be a valuation. Then:

$(1)$
${|\Gamma, A\to B \Rightarrow \Delta|_{M,\nu} \geqslant 0}$  iff,\,
for every ${r \in (-\infty, 1]}$, at least one of the inequalities
\,${|\Gamma, \SpV{p} \Rightarrow \Delta|_{M[\SpV{p}\mapsto r],\, \nu} \geqslant 0}$ 
or 
\,${|B \Rightarrow \SpV{p}, A|_{M[\SpV{p}\mapsto r],\, \nu} \geqslant 0}$
holds;

$(2)$
${|\Gamma \Rightarrow A\to B, \Delta|_{M,\nu} \geqslant 0}$  iff
\,${|\Gamma \Rightarrow \Delta|_{M, \nu} \geqslant 0}$ 
and \,${|\Gamma,A \Rightarrow B,\Delta|_{M, \nu} \geqslant 0}$;

$(3)$
${|\Gamma \Rightarrow \forall x A, \Delta|_{M,\nu} \geqslant 0}$  iff
\,${|\Gamma \Rightarrow \RePl{A}{x}{y}, \Delta|_{M, \nu[y\mapsto d]} \geqslant 0}$
for every ${d\in\mathcal{D}}$;

$(4)$
${|\Gamma, \exists x A \Rightarrow\Delta|_{M,\nu} \geqslant 0}$  iff
\,${|\Gamma, \RePl{A}{x}{y} \Rightarrow\Delta|_{M, \nu[y\mapsto d]} \geqslant 0}$
for every ${d\in\mathcal{D}}$;

$(5)$
${|\Gamma, \forall x A \Rightarrow\Delta|_{M,\nu} \geqslant 0}$ if
\,${|\Gamma, \RePl{A}{x}{y} \Rightarrow\Delta|_{M, \nu[y\mapsto d]} \geqslant 0}$
for some ${d\in\mathcal{D}}$;

$(6)$
${|\Gamma \Rightarrow \exists x A, \Delta|_{M,\nu} \geqslant 0}$ if
\,${|\Gamma \Rightarrow \RePl{A}{x}{y}, \Delta|_{M, \nu[y\mapsto d]} \geqslant 0}$
for some ${d\in\mathcal{D}}$;

$(\widetilde{5})$
${|\Gamma, B \Rightarrow \Delta|_{M,\nu} \geqslant 0}$  iff,\,
for every ${r \in (-\infty, 1]}$, 
at least one of the inequalities
\,${|\Gamma, \SpV{p} \Rightarrow \Delta|_{M[\SpV{p}\mapsto r],\, \nu} \geqslant 0}$ 
or 
\,${|B \Rightarrow \SpV{p}|_{M[\SpV{p}\mapsto r],\, \nu} \geqslant 0}$
holds;

$(\widetilde{6})$
${|\Gamma \Rightarrow B, \Delta|_{M,\nu} \geqslant 0}$  iff,\,
for every ${r \in [0, +\infty)}$, 
at least one of the inequalities
\,${|\Gamma \Rightarrow \SpV{q}, \Delta|_{M[\SpV{q}\mapsto r],\, \nu} \geqslant 0}$ 
or 
\,${|\SpV{q} \Rightarrow B|_{M[\SpV{q}\mapsto r],\, \nu} \geqslant 0}$
holds.
\end{lemma}

\begin{proof}
Assertions (1)--(6) stated above are proved similarly to 
assertions (1)--(6) in \cite[Lemma~2]{Ger2017}. 

Let us prove assertions $(\widetilde{5})$ and $(\widetilde{6})$.
Denote $|\Gamma|_{M,\nu}$, $|\Delta|_{M,\nu}$, and $|B|_{M,\nu}$
by $\gamma$, $\delta$, and $b$, respectively; 
and notice that ${0 \leqslant b \leqslant 1}$.

Assertion $(\widetilde{5})$ is equivalent to the following: 
$$(\widetilde{5}')\ \delta-\gamma+1<b  \iff  
  (\widetilde{5}'')\ \delta-\gamma+1<r<b
        \,\text{ for some }  r\leqslant 1.$$
It is clear that $(\widetilde{5}'')$ implies $(\widetilde{5}')$.
If $(\widetilde{5}')$ holds, then by the density of the set of all real numbers, 
both inequalities from $(\widetilde{5}'')$ hold for some ${r < b \leqslant 1}$.
Thus $(\widetilde{5})$ holds.

Assertion $(\widetilde{6})$ is equivalent to the following: 
$$b<\gamma-\delta+1  \iff  
  b<r<\gamma-\delta+1  \text{ for some }  r\geqslant 0.$$
By the density of the set of all real numbers, the last equivalence 
holds and so does $(\widetilde{6})$.
\end{proof}

\begin{theorem}[soundness of $\GthreeLukA$]  \label{GTh:CorrGthree}
If \ ${\vdash_{\GthreeLukA} \mathcal{H}}$, \,then
\ ${\vDash \mathcal{H}}$.
\end{theorem}

\begin{proof}
All axioms of $\GthreeLukA$ are obviously valid, and all the inference rules 
of $\GthreeLukA$ are sound by Lemma~\ref{GLem:CorrSemInvGthreeRules}.
\end{proof}

Using the semantical invertibility of the propositional rules of $\GthreeLukA$ 
(see Lemma \ref{GLem:CorrSemInvGthreeRules}), we can easily prove

\begin{proposition}
Let $\mathcal{H}$ be a quantifier-free hypersequent. 
If \ ${\vDash \mathcal{H}}$, \,then
\ ${\vdash_{\GthreeLukA} \mathcal{H}}$.
\end{proposition}

\section{Invertibility of the rules of the calculus $\GthreeLukA$
         and~its~relationship~to~the~calculus~$\GoneLukA$}
\label{Sec:Invert}

Suppose $\mathfrak{C}$ is a calculus.
By $h(D)$ we denote the height of a (tree-like) $\mathfrak{C}$-proof $D$. 
Let us recall some definitions (cf., e.g., \cite[Section 3.4.4]{Troelstra2000}).

A rule is called \emph{admissible} for $\mathfrak{C}$ if, for 
all applications $\mathcal{H}_1;\ldots;\mathcal{H}_k / \mathcal{H}$
of the rule and all $\mathfrak{C}$-proofs
$D_1$ of $\mathcal{H}_1$, \ldots, $D_k$ of $\mathcal{H}_k$,
there exists a $\mathfrak{C}$-proof $D$ of $\mathcal{H}$;
the rule is called \emph{hp-admissible}, or \emph{height-preserving admissible},
for $\mathfrak{C}$ if, in addition, the condition
${h(D) \leqslant \max\{h(D_1),\ldots,h(D_k)\}}$ holds.
Everywhere in the sequel, the existense of such a proof $D$ means that 
it can be constructed if such proofs $D_1,\ldots,$ $D_k$ are given.

A $k$-premise rule $\mathcal{R}$ is called \emph{invertible} 
(resp. \emph{hp-invertible}, or \emph{height-preserving invertible})
in $\mathfrak{C}$ if, for each ${i=1,\ldots,k}$,\, the rule
${\{\langle \mathcal{H}, \mathcal{H}_i \rangle \mid 
  \langle \mathcal{H}_1,\ldots,\mathcal{H}_k,\mathcal{H} \rangle \in \mathcal{R}\}}$
is admissible (resp. hp-admissible) for $\mathfrak{C}$.

\begin{lemma}  \label{GLem:AtEwSplitAdmGthree}
The following rules are hp-admissible for the calculus $\GthreeLukA$:
$$
\dfrac{\mathcal{G}}{\mathcal{G}\,|\,S}~\text{\textnormal{(ew)}}^3,
\quad
\dfrac{\mathcal{G}\,|\,\Gamma_1,\Gamma_2\Rightarrow\Delta_1,\Delta_2}
  {\mathcal{G}\,|\,\Gamma_1\Rightarrow\Delta_1\,|\,\Gamma_2\Rightarrow\Delta_2}~\text{\textnormal{(split)}}^3,
\quad
\dfrac{\mathcal{G}\,|\,\Gamma\Rightarrow\Delta}
  {\mathcal{G}\,|\,\Gamma,P \Rightarrow P,\Delta}~\text{\textnormal{(at$\Rightarrow$at)}}^3,
$$
where $P$ is an atom \textnormal{(}i.e., an atomic $\RPLA$-formula or
a semipropositional variable\textnormal{)}.
\end{lemma}

\begin{proof}
1. The rule $\text{(ew)}^3$ is obviously hp-admissible: 
if $D$ is a proof of a premise of the rule, 
then ${D\,|\,S}$ is a proof of its conclusion.

2. Let us establish the hp-admissibility of the rule $\text{(split)}^3$.

Suppose $\mathcal{S}_0$ is a sequent occurrence in the root of a proof search 
tree $D_0$; 
then we say that an ancestor $\mathcal{S}$ of the occurrence $\mathcal{S}_0$ 
is \emph{augmentable}  
unless $\mathcal{S}$ is an ancestor of an occurrence $\mathcal{S}'$ 
of a sequent $S'$ such that:

(i) $S'$ has the form
\,(a)~${B \Rightarrow \SpV{p}, A}$,\,
\,(b)~${\forall x A \Rightarrow \SpV{p}}$\, or
\,${\RePl{A}{x}{t} \Rightarrow \SpV{p}}$,\, or
\,(c)~${\SpV{q} \Rightarrow \exists x A}$\, or
\,${\SpV{q} \Rightarrow \RePl{A}{x}{t}}$;\,

(ii) in $D_0$,\,
$\mathcal{S}'$ is a sequent occurrence in the premise of an application
of the rule
(a)~${(\to\Rightarrow)^3}$,
(b)~${(\forall\Rightarrow)^3}$, or
(c)~${(\Rightarrow\exists)^3}$, respectively;\, and 

(iii) $\mathcal{S}'$ is distinguished in the formulation of this rule.

Let $D_0$ be a proof for the premise
\,${\mathcal{G}\,|\,\Gamma_1,\Gamma_2\Rightarrow\Delta_1,\Delta_2}$\, 
of the rule $\text{(split)}^3$.
A tree $D$ is constructed from $D_0$ as follows:
each occurrence $\mathcal{S}$ of a sequent $S$ of the form
\,$\Pi_1,\Pi_2\Rightarrow$ $\Sigma_1,\Sigma_2$\, such that

(1) $\mathcal{S}$ is an augmentable ancestor of the distinguished occurrence 
of the sequent \,${\Gamma_1,\Gamma_2\Rightarrow\Delta_1,\Delta_2}$\,  
in the root of $D_0$, and

(2) for each ${i=1,2}$\, and 
each formula occurrence $\mathcal{F}$ (as a sequent member) in $S$,
if $\mathcal{F}$ is contained in the distinguished occurrence $\Pi_i$ or 
$\Sigma_i$ in $S$, 
then $\mathcal{F}$ is an ancestor of some for\-mula occurrence contained in
the distinguished occurrence $\Gamma_i$ or $\Delta_i$ in the root of $D_0$,\\
is replaced by \,${\Pi_1\Rightarrow\Sigma_1 \,|\, \Pi_2\Rightarrow\Sigma_2}$.

The rules of $\GthreeLukA$ guarantee that, in a premise of a rule application,
there is exactly one augmentable ancestor of the principal sequent occurrence.
Therefore, when the tree $D$ is constructed,
exactly one sequent occurrence in each node hypersequent of the proof $D_0$
is split into two sequents.
Then it is easy to see that each application of a rule in $D_0$ 
is turned into an application of the same rule.
Clearly, the hypersequent 
\,${\mathcal{G}\,|\,\Gamma_1\Rightarrow\Delta_1\,|\,\Gamma_2\Rightarrow\Delta_2}$\,
is in the root of the tree $D$.
Hence $D$ is a proof search tree for the conclusion 
of the rule $\text{(split)}^3$.

Let $\mathcal{L}$ be a leaf of the tree $D_0$.
Let $\mathcal{S}$ be an occurrence of an atomic sequent $S$ in $\mathcal{L}$ 
such that $S$ has the form \,${\Pi_1,\Pi_2\Rightarrow\Sigma_1,\Sigma_2}$,\,
and $\mathcal{S}$ and $S$ meet conditions (1) and (2) above.  Then
the leaf of the tree $D$ obtained from $\mathcal{L}$ 
contains the atomic sequents
\,${\Pi_1\Rightarrow\Sigma_1}$\, and \,${\Pi_2\Rightarrow\Sigma_2}$.\,
So $D$ is a proof.

It remains to note that ${h(D) = h(D_0)}$.

3. The rule $\text{(at$\Rightarrow$at)}^3$ is hp-admissible, since,
given a proof $D$ for \,${\mathcal{G}\,|\,\Gamma\Rightarrow\Delta}$,\,
we can construct a proof for \,${\mathcal{G}\,|\,\Gamma,P \Rightarrow P,\Delta}$\, 
(with $P$ being an atom) in the following way.
First, in $D$, rename 
all proper semipropositional variables and proper parameters of $D$ 
that occur in $P$ to new distinct ones. 
Next, in the resulting proof for \,${\mathcal{G}\,|\,\Gamma\Rightarrow\Delta}$,\,
add the atom $P$ to the antecedent and succedent of each augmentable 
ancestor of the distinguished occurrence of the sequent 
\,${\Gamma\Rightarrow\Delta}$\, in the root.
\end{proof}

\begin{lemma}  \label{GLem:InvGthreeRules}
All the inference rules of the calculus $\GthreeLukA$ are hp-invertible in it.
\end{lemma}

\begin{proof}
The rule ${(\forall\Rightarrow)^3}$ is hp-invertible, 
since we can obtain its premise from its conclusion using rules, 
which are hp-admissible (by Lemma \ref{GLem:AtEwSplitAdmGthree}):
\begin{center}
\def\ScoreOverhang{0pt}
\AxiomC{$\mathcal{G}\,|\, \Gamma, \forall x A \Rightarrow \Delta$}
\RightLabel{$\text{(at$\Rightarrow$at)}^3$}
\UnaryInfC{$\mathcal{G}\,|\, \Gamma, \forall x A, \SpV{p} \Rightarrow \SpV{p}, \Delta$}
\RightLabel{$\text{(split)}^3$}
\UnaryInfC{$\mathcal{G}\,|\, \Gamma, \SpV{p} \Rightarrow \Delta \,|\, \forall x A \Rightarrow \SpV{p}$}
\RightLabel{$\text{(ew)}^3$.}
\UnaryInfC{$\mathcal{G}\,|\, \Gamma, \SpV{p} \Rightarrow\Delta \,|\, \forall x A \Rightarrow \SpV{p} \,|\, \RePl{A}{x}{t} \Rightarrow \SpV{p}$}
\DisplayProof
\end{center}

The hp-invertibility of the rule ${(\Rightarrow\exists)^3}$ 
is established very similarly.

The fact that all the inference rules of $\GthreeLukA$ are repetition-free
allows us to demonstrate the hp-invertibility of the rules 
${(\to\Rightarrow)^3}$, ${(\Rightarrow\to)^3}$,
${(\Rightarrow\forall)^3}$, and ${(\exists\Rightarrow)^3}$
according to the classical scheme (see, e.g., 
\cite[Proposition 3.5.4]{Troelstra2000}).
We give these demonstrations in full because later\footnote{
  See the proofs of Lemma \ref{GLem:StepRearrangingGthreeProof} and
  Theorem \ref{GTh:MidHsGthree}.
}
we will need to check that formal proofs constructed in them 
enjoy some properties.

\smallskip\noindent\textbf{I.}
Let us demonstrate that the rule ${(\to\Rightarrow)^3}$ is hp-invertible.
Toward this end, we show that, given a proof $D$ for a hypersequent of the form
\,${\mathcal{G}\,|\,\Gamma, A\to B \Rightarrow \Delta}$\,
and a semipropositional variable $\SpV{p}$ not occurring in the hypersequent,
we can construct a proof $D'$ for
\,${\mathcal{G}\,|\, \Gamma, \SpV{p} \Rightarrow \Delta\,|\,B \Rightarrow \SpV{p}, A}$\,
with \,${h(D') \leqslant h(D)}$.
We proceed by induction on $h(D)$.

We can assume that $\SpV{p}$ does not occur in $D$
(otherwise replace all occurrences of $\SpV{p}$ in $D$ by 
a semipropositional variable of type~1 not occurring in $D$).

1. If ${h(D)=0}$ (i.e., $D$ consists of a single axiom), then
$\mathcal{G}$ is an axiom, hence so is
\,${\mathcal{G}\,|\, \Gamma, \SpV{p} \Rightarrow \Delta\,|\,B \Rightarrow \SpV{p}, A}$.

2. Let the root hypersequent
\,${\mathcal{G}\,|\,\Gamma, A\to B \Rightarrow \Delta}$\, in $D$ 
be the conclusion of an application $R$ of a rule $\mathcal{R}$.

2.1. Suppose 
the principal formula occurrence in $R$ is the distinguished occurrence of ${A\to B}$.
By $D_1$ denote the subtree of the root of $D$;
$D_1$ is a proof for the premise of $R$.
The premise has the form 
\,${\mathcal{G}\,|\, \Gamma, \SpV{p}_1 \Rightarrow \Delta\,|\,B \Rightarrow \SpV{p}_1, A}$.
Then replacing all occurrences of $\SpV{p}_1$ in $D_1$ by $\SpV{p}$ 
yields a proof $D'$ for
\,${\mathcal{G}\,|\, \Gamma, \SpV{p} \Rightarrow \Delta\,|\,B \Rightarrow \SpV{p}, A}$\,
with ${h(D') < h(D)}$. 

2.2. Now suppose 
the principal formula occurrence in $R$ is not the distinguished occurrence of ${A\to B}$.

2.2.1. If $\mathcal{R}$ is a one-premise rule, the proof $D$ looks like this:
\begin{center}
\def\ScoreOverhang{0pt}
\AxiomC{$D_1$} \noLine 
\UnaryInfC{$\mathcal{G}_1 \,|\, \Gamma_1, A \to B \Rightarrow \Delta_1$}
\RightLabel{\,$\mathcal{R}$.}
\UnaryInfC{$\mathcal{G} \,|\, \Gamma, A \to B \Rightarrow \Delta$}
\DisplayProof
\end{center}

By applying the induction hypothesis to the proof $D_1$,
we construct a proof $D_1'$ for
\,${\mathcal{G}_1 \,|\, \Gamma_1, \SpV{p} \Rightarrow \Delta_1 \,|\, B \Rightarrow \SpV{p}, A}$\,
with \,${h(D_1') \leqslant h(D_1)}$.
By applying $\mathcal{R}$ to the root hypersequent of the proof $D_1'$, 
we obtain a proof $D'$ for
\,${\mathcal{G}\,|\, \Gamma, \SpV{p} \Rightarrow \Delta\,|\,B \Rightarrow \SpV{p}, A}$\,
such that ${h(D') \leqslant h(D)}$.

2.2.2. If $\mathcal{R}$ is a two-premise rule,
i.e., the rule ${(\Rightarrow\to)^3}$, then the proof $D$ looks like this:
\begin{center}
\def\ScoreOverhang{0pt}
\AxiomC{$D_1$} \noLine 
\UnaryInfC{$\mathcal{G}_1 \,|\, \Gamma_1, A \to B \Rightarrow \Delta_1$}
\AxiomC{$D_2$} \noLine 
\UnaryInfC{$\mathcal{G}_2 \,|\, \Gamma_2, A \to B \Rightarrow \Delta_2$}
\RightLabel{\,$\mathcal{R}$.}
\BinaryInfC{$\mathcal{G} \,|\, \Gamma, A \to B \Rightarrow \Delta$}
\DisplayProof
\end{center}

For each ${i=1,2}$, by the induction hypothesis applied to the proof $D_i$,
we construct a proof $D_i'$ for
\,${\mathcal{G}_i \,|\, \Gamma_i, \SpV{p} \Rightarrow \Delta_i \,|\, B \Rightarrow \SpV{p}, A}$\,
with \,${h(D_i') \leqslant h(D_i)}$.

By applying $\mathcal{R}$ to the root hypersequents of 
the proofs $D_1'$ and $D_2'$, we get a proof $D'$ for
\,${\mathcal{G} \,|\, \Gamma, \SpV{p} \Rightarrow \Delta \,|\, B \Rightarrow \SpV{p}, A}$\,
with ${h(D') \leqslant h(D)}$.

\smallskip\noindent\textbf{II.}
In order to establish the hp-invertibility of the rule ${(\Rightarrow\to)^3}$,
we show that, given a proof $D$ for a hypersequent of the form
\,${\mathcal{G}\,|\,\Gamma \Rightarrow A\to B, \Delta}$,\, we can construct 
a proof $D'$ for \,${\mathcal{G}\,|\, \Gamma \Rightarrow \Delta}$\,
and a proof $D''$ for \,${\mathcal{G}\,|\, \Gamma, A \Rightarrow B, \Delta}$\,
such that ${h(D') \leqslant h(D)}$ and ${h(D'') \leqslant h(D)}$.
We use induction on $h(D)$.

1. If ${h(D)=0}$, then $\mathcal{G}$ is an axiom, and so are
\,${\mathcal{G}\,|\, \Gamma \Rightarrow \Delta}$\, and
\,${\mathcal{G}\,|\, \Gamma, A \Rightarrow B, \Delta}$.

2. Let the root hypersequent
\,${\mathcal{G}\,|\,\Gamma \Rightarrow A\to B, \Delta}$\, in $D$ 
be the conclusion of an application $R$ of a rule $\mathcal{R}$.

2.1. If 
the principal formula occurrence in $R$ is the distinguished occurrence of ${A\to B}$,
then the subtrees of the root of $D$ are the desired proofs.

2.2. Suppose 
the principal formula occurrence in $R$ is not the distinguished occurrence of ${A\to B}$.

2.2.1. In the case the rule $\mathcal{R}$ is one-premise, 
the proof $D$ looks like this:
\begin{center}
\def\ScoreOverhang{0pt}
\AxiomC{$D_1$} \noLine 
\UnaryInfC{$\mathcal{G}_1 \,|\, \Gamma_1 \Rightarrow A \to B, \Delta_1$}
\RightLabel{\,$\mathcal{R}$.}
\UnaryInfC{$\mathcal{G} \,|\, \Gamma \Rightarrow A \to B, \Delta$}
\DisplayProof
\end{center}

Using the induction hypothesis, from the proof $D_1$, 
we construct a proof $D_1'$ for
\,${\mathcal{G}_1 \,|\, \Gamma_1 \Rightarrow \Delta_1}$\, and
a proof $D_1''$ for
\,${\mathcal{G}_1 \,|\, \Gamma_1, A \Rightarrow B, \Delta_1}$\,
such that ${h(D_1') \leqslant h(D_1)}$ and ${h(D_1'') \leqslant h(D_1)}$.

Applying $\mathcal{R}$ to the root hypersequent of the proof $D_1'$ gives
a proof $D'$ for \,${\mathcal{G}\,|\, \Gamma \Rightarrow \Delta}$\, 
with \,${h(D') \leqslant h(D)}$,\,
and applying $\mathcal{R}$ to the root hypersequent of the proof $D_1''$ gives
a proof $D''$ for \,${\mathcal{G}\,|\, \Gamma, A \Rightarrow B, \Delta}$\,
with \,${h(D'') \leqslant h(D)}$.

2.2.2. In the case the rule $\mathcal{R}$ is two-premise, 
the proof $D$ looks like this:
\begin{center}
\def\ScoreOverhang{0pt}
\AxiomC{$D_1$} \noLine 
\UnaryInfC{$\mathcal{G}_1 \,|\, \Gamma_1 \Rightarrow A \to B, \Delta_1$}
\AxiomC{$D_2$} \noLine 
\UnaryInfC{$\mathcal{G}_2 \,|\, \Gamma_2 \Rightarrow A \to B, \Delta_2$}
\RightLabel{\,$\mathcal{R}$.}
\BinaryInfC{$\mathcal{G} \,|\, \Gamma \Rightarrow A \to B, \Delta$}
\DisplayProof
\end{center}

For each ${i=1,2}$, by the induction hypothesis applied to the proof $D_i$,
we construct a proof $D_i'$ for 
\,${\mathcal{G}_i \,|\, \Gamma_i \Rightarrow \Delta_i}$\, and
a proof $D_i''$ for
\,${\mathcal{G}_i \,|\, \Gamma_i, A \Rightarrow B, \Delta_i}$\,
such that ${h(D_i') \leqslant h(D_i)}$ and ${h(D_i'') \leqslant h(D_i)}$.

Next, by applying $\mathcal{R}$ to the root hypersequents of 
the proofs $D_1'$ and $D_2'$, we obtain a proof $D'$ for 
\,${\mathcal{G}\,|\, \Gamma \Rightarrow \Delta}$\, 
with ${h(D') \leqslant h(D)}$.

Finally, applying $\mathcal{R}$ to the root hypersequents of 
the proofs $D_1''$ and $D_2''$ yields a proof $D''$ for
\,${\mathcal{G}\,|\, \Gamma, A \Rightarrow B, \Delta}$\,
with ${h(D'') \leqslant h(D)}$.

\smallskip\noindent\textbf{III.}
To establish the hp-invertibility of the rule ${(\Rightarrow\forall)^3}$,
we show that, given a proof $D$ for a hypersequent of the form
\,${\mathcal{G}\,|\,\Gamma \Rightarrow \forall x A, \Delta}$\,
and a parameter $a$ not occurring in the hypersequent,
we can construct a proof $D'$ for
\,${\mathcal{G}\,|\, \Gamma \Rightarrow \RePl{A}{x}{a}, \Delta}$\,
with \,${h(D') \leqslant h(D)}$.
This is done by induction on $h(D)$. 

We can assume that $a$ does not occur in $D$
(otherwise replace all occurrences of $a$ in $D$ by 
a parameter not occurring in $D$).

1. If ${h(D)=0}$, then $\mathcal{G}$ is an axiom and so is
\,${\mathcal{G}\,|\, \Gamma \Rightarrow \RePl{A}{x}{a}, \Delta}$.

2. Let the root hypersequent
${\mathcal{G}\,|\,\Gamma \Rightarrow \forall x A, \Delta}$ in $D$ 
be the conclusion of an application $R$ of a rule $\mathcal{R}$.

2.1. Suppose 
the principal formula occurrence in $R$ is the distinguished occurrence of $\forall x A$.
By $D_1$ denote the subtree of the root of $D$;
$D_1$ is a proof for the premise of $R$.
The premise has the form 
\,${\mathcal{G}\,|\, \Gamma \Rightarrow \RePl{A}{x}{a_1}, \Delta}$.
By replacing all occurrences of $a_1$ in $D_1$ by $a$, 
we get a proof $D'$ for
\,${\mathcal{G}\,|\, \Gamma \Rightarrow \RePl{A}{x}{a}, \Delta}$\,
with ${h(D') < h(D)}$. 

2.2. Next, suppose 
the principal formula occurrence in $R$ is not the distinguished occurrence of $\forall x A$.

2.2.1. If $\mathcal{R}$ is one-premise, the proof $D$ looks like this:
\begin{center}
\def\ScoreOverhang{0pt}
\AxiomC{$D_1$} \noLine 
\UnaryInfC{$\mathcal{G}_1 \,|\, \Gamma_1 \Rightarrow \forall x A, \Delta_1$}
\RightLabel{\,$\mathcal{R}$.}
\UnaryInfC{$\mathcal{G} \,|\, \Gamma \Rightarrow \forall x A, \Delta$}
\DisplayProof
\end{center}

Using the induction hypothesis, we transform the proof $D_1$ into 
a proof $D_1'$ for
\,${\mathcal{G}_1 \,|\, \Gamma_1 \Rightarrow \RePl{A}{x}{a}, \Delta_1}$\,
such that ${h(D_1') \leqslant h(D_1)}$.
By applying $\mathcal{R}$ to the root hypersequent of the proof $D_1'$,
we have a proof $D'$ for
\,${\mathcal{G}\,|\, \Gamma \Rightarrow \RePl{A}{x}{a}, \Delta}$\,
with ${h(D') \leqslant h(D)}$.

2.2.2. If $\mathcal{R}$ is two-premise, the proof $D$ looks like this:
\begin{center}
\def\ScoreOverhang{0pt}
\AxiomC{$D_1$} \noLine 
\UnaryInfC{$\mathcal{G}_1 \,|\, \Gamma_1 \Rightarrow \forall x A, \Delta_1$}
\AxiomC{$D_2$} \noLine 
\UnaryInfC{$\mathcal{G}_2 \,|\, \Gamma_2 \Rightarrow \forall x A, \Delta_2$}
\RightLabel{\,$\mathcal{R}$.}
\BinaryInfC{$\mathcal{G} \,|\, \Gamma \Rightarrow \forall x A, \Delta$}
\DisplayProof
\end{center}

For each ${i=1,2}$, by the induction hypothesis, 
we transform the proof $D_i$ into a proof $D_i'$ for
\,${\mathcal{G}_i \,|\, \Gamma_i \Rightarrow \RePl{A}{x}{a}, \Delta_i}$\,
such that ${h(D_i') \leqslant h(D_i)}$.

By applying $\mathcal{R}$ to the root hypersequents of 
the proofs $D_1'$ and $D_2'$, we obtain a proof $D'$ for
\,${\mathcal{G}\,|\, \Gamma \Rightarrow \RePl{A}{x}{a}, \Delta}$\, 
with \,${h(D') \leqslant h(D)}$.

\smallskip\noindent\textbf{IV.}
The hp-invertibility of the rule ${(\exists\Rightarrow)^3}$ 
is established very similarly to
the hp-invertibility of the rule ${(\Rightarrow\forall)^3}$, see item~III.
\end{proof}

\noindent\textbf{Remark 1.}
We know the following about whether the inference rules of 
the calculus $\GtwoLukA$ are invertible in it.
The rules ${(\forall\Rightarrow)^2}$ and ${(\Rightarrow\exists)^2}$ are 
hp-invertible because, for each of them, its premise includes its conclusion.
Using arguments like those given in the proof of 
Lemma \ref{GLem:InvGthreeRules}, we can establish the hp-invertibility of
the rules ${(\Rightarrow\to)^2}$, ${(\Rightarrow\forall)^2}$, and
${(\exists\Rightarrow)^2}$.
However, we do not know whether the rule ${(\to\Rightarrow)^2}$ is invertible.

\begin{lemma} \label{GLem:ecAdmGthree}
The following rule is hp-admissible for the calculus $\GthreeLukA$:
$$\dfrac{\mathcal{G}\,|\,S\,|\,S}{\mathcal{G}\,|\,S}~\text{\textnormal{(ec)}}^3.$$
\end{lemma}

\begin{proof}
We show that a proof $D$ for a hypersequent of the form
\,${\mathcal{G}\,|\,S\,|\,S}$\, can be transformed into 
a proof $\widehat{D}$ for
\,${\mathcal{G}\,|\,S}$\, with ${h(\widehat{D}) \leqslant h(D)}$.
We proceed by induction on $h(D)$.

1. If ${h(D)=0}$, then the hypersequents \,${\mathcal{G}\,|\,S\,|\,S}$\, and
\,${\mathcal{G}\,|\,S}$\, are axioms.

2. Let the root hypersequent \,${\mathcal{G}\,|\,S\,|\,S}$\, in $D$ 
be the conclusion of an application $R$ of a rule $\mathcal{R}$.

2.1. If the principal sequent occurrence in $R$ is not one of 
the two occurrences of $S$ distinguished in \,${\mathcal{G}\,|\,S\,|\,S}$,\,
then we apply the induction hypothesis to the proof for each premise of $R$
and next use $\mathcal{R}$ to obtain the desired proof for
\,${\mathcal{G}\,|\,S}$.

2.2. Otherwise, we are to treat each inference rule of $\GthreeLukA$ 
as $\mathcal{R}$.
However, all these cases are similar to one another. 
So we treat only the case where $\mathcal{R}$ is ${(\forall\Rightarrow)^3}$.
Then the proof $D$ has the form: 
\begin{center}
\def\ScoreOverhang{0pt}
\AxiomC{$D_1$} \noLine 
\UnaryInfC{$\mathcal{G} \,|\, \Gamma, \SpV{p} \Rightarrow \Delta \,|\, \forall x A \Rightarrow \SpV{p} \,|\, \RePl{A}{x}{t} \Rightarrow \SpV{p} \,|\, \Gamma, \forall x A \Rightarrow \Delta$}
\RightLabel{$(\forall\Rightarrow)^3$.}
\UnaryInfC{$\mathcal{G} \,|\, \Gamma, \forall x A \Rightarrow \Delta \,|\, \Gamma, \forall x A \Rightarrow \Delta$}
\DisplayProof
\end{center}

Since the rule ${(\forall\Rightarrow)^3}$ is hp-invertible
(see Lemma \ref{GLem:InvGthreeRules}), 
given the proof $D_1$, we can find a proof $D_1'$ for
$$\mathcal{G} \,|\, \Gamma, \SpV{p} \Rightarrow \Delta \,|\, \forall x A \Rightarrow \SpV{p} \,|\, \RePl{A}{x}{t} \Rightarrow \SpV{p} \,|\, \Gamma, \SpV{p}_1 \Rightarrow \Delta \,|\, \forall x A \Rightarrow \SpV{p}_1 \,|\, \RePl{A}{x}{t} \Rightarrow \SpV{p}_1,$$
where $\SpV{p}_1$ does not occur in the root hypersequent of $D_1$ 
and ${h(D_1') \leqslant h(D_1)}$.

Replacing all occurrences of $\SpV{p}_1$ in $D_1'$ by $\SpV{p}$ yields
a proof $D_1''$ for
$$\mathcal{G} \,|\, \Gamma, \SpV{p} \Rightarrow \Delta \,|\, \forall x A \Rightarrow \SpV{p} \,|\, \RePl{A}{x}{t} \Rightarrow \SpV{p} \,|\, \Gamma, \SpV{p} \Rightarrow \Delta \,|\, \forall x A \Rightarrow \SpV{p} \,|\, \RePl{A}{x}{t} \Rightarrow \SpV{p}\,;$$
whence using the induction hypothesis three times, 
we get a proof $\widetilde{D}_1$ for
$$\mathcal{G} \,|\ \Gamma, \SpV{p} \Rightarrow \Delta \,|\, \forall x A \Rightarrow \SpV{p} \,|\, \RePl{A}{x}{t} \Rightarrow \SpV{p}$$
such that ${h(\widetilde{D}_1) \leqslant h(D_1'') \leqslant h(D_1)}$.

Finally, by applying ${(\forall\Rightarrow)^3}$ to the root hypersequent of 
the proof $\widetilde{D}_1$, we obtain the desired proof $\widehat{D}$ for
\,${\mathcal{G} \,|\, \Gamma, \forall x A \Rightarrow \Delta}$\,
with ${h(\widehat{D}) \leqslant h(D)}$.
\end{proof}

\begin{lemma}  \label{GLem:GoneRulesAdmInGthree}
Each inference rule of the calculus $\GoneLukA$ is admissible for 
the calculus $\GthreeLukA$.
\end{lemma}

\begin{proof}
An application of the rule ${(\to\Rightarrow)^1}$, ${(\Rightarrow\to)^1}$,
${(\Rightarrow\forall)^1}$, or ${(\exists\Rightarrow)^1}$ of $\GoneLukA$ 
can be represented as an application of the corresponding rule of $\GthreeLukA$ 
followed by an appplication of the rule $\text{\textnormal{(ec)}}^3$.
E.g., an application of ${(\Rightarrow\to)^1}$ is represented thus:
\begin{center}
\def\ScoreOverhang{0pt}
\AxiomC{$\mathcal{G} \,|\, \Gamma \Rightarrow A\to B, \Delta \,|\, \Gamma \Rightarrow \Delta$;}
\AxiomC{$\mathcal{G} \,|\, \Gamma \Rightarrow A\to B, \Delta \,|\, \Gamma, A \Rightarrow B, \Delta$}
\RightLabel{$(\Rightarrow\to)^3$}
\BinaryInfC{$\mathcal{G} \,|\, \Gamma \Rightarrow A\to B, \Delta \,|\, \Gamma \Rightarrow A\to B, \Delta$}
\RightLabel{$\text{\textnormal{(ec)}}^3$.}
\UnaryInfC{$\mathcal{G} \,|\, \Gamma \Rightarrow A\to B, \Delta$}
\DisplayProof
\end{center}
By Lemma \ref{GLem:ecAdmGthree}, the rule $\text{\textnormal{(ec)}}^3$ 
is admissible for $\GthreeLukA$.
So these four rules of $\GoneLukA$ are admissible for $\GthreeLukA$.

The rule ${(\forall\Rightarrow)^1}$ is admissible for $\GthreeLukA$, since
an application of it can be represented as several applications of rules,
which are admissible for $\GthreeLukA$ 
(by Lemmas \ref{GLem:AtEwSplitAdmGthree} and \ref{GLem:ecAdmGthree}),
as follows:
\begin{center}
\def\ScoreOverhang{0pt}
\AxiomC{$\mathcal{G} \,|\, \Gamma, \forall x A \Rightarrow \Delta \,|\, \Gamma, \RePl{A}{x}{t} \Rightarrow \Delta$}
\RightLabel{$\text{(at$\Rightarrow$at)}^3$}
\UnaryInfC{$\mathcal{G} \,|\, \Gamma, \forall x A \Rightarrow \Delta \,|\, \Gamma, \RePl{A}{x}{t}, \SpV{p} \Rightarrow \SpV{p}, \Delta$}
\RightLabel{$\text{(split)}^3$}
\UnaryInfC{$\mathcal{G} \,|\, \Gamma, \forall x A \Rightarrow \Delta \,|\, \Gamma, \SpV{p} \Rightarrow \Delta \,|\, \RePl{A}{x}{t} \Rightarrow \SpV{p}$}
\RightLabel{$\text{(ew)}^3$}
\UnaryInfC{$\mathcal{G} \,|\, \Gamma, \forall x A \Rightarrow \Delta \,|\, \Gamma, \SpV{p} \Rightarrow\Delta \,|\, \forall x A \Rightarrow \SpV{p} \,|\, \RePl{A}{x}{t} \Rightarrow \SpV{p}$}
\RightLabel{$(\forall\Rightarrow)^3$}
\UnaryInfC{$\mathcal{G} \,|\, \Gamma, \forall x A \Rightarrow \Delta \,|\, \Gamma, \forall x A \Rightarrow \Delta$}
\RightLabel{$\text{\textnormal{(ec)}}^3$,}
\UnaryInfC{$\mathcal{G} \,|\, \Gamma, \forall x A \Rightarrow \Delta$}
\DisplayProof
\end{center}
where $\SpV{p}$ does not occur in the top hypersequent.

The rule ${(\Rightarrow\exists)^1}$ is treated similarly to
${(\forall\Rightarrow)^1}$.
\end{proof}

\begin{theorem}  \label{GTh:GOneProvGthreeProv}
Suppose $\mathcal{H}$ is a hypersequent of the calculus $\GoneLukA$.
If \ ${\vdash_{\GoneLukA} \mathcal{H}}$, then
\,${\vdash_{\GthreeLukA} \mathcal{H}}$.
\end{theorem}

\begin{proof}
All axioms of $\GoneLukA$ are axioms of $\GthreeLukA$,
and all the inference rules of $\GoneLukA$ are admissible for $\GthreeLukA$ 
by Lemma~\ref{GLem:GoneRulesAdmInGthree}.
\end{proof}

\section{Transforming $\GthreeLukA$-proofs according to tactics} 
\label{Sec:Transform}

As in \cite[Section 4.3]{Ger2017}, 
to organize bottom-up $\GthreeLukA$-proof search, 
we can use an auxiliary algorithm $\mathfrak{t}$,
called a (\emph{proof search}) \emph{tactic}, 
that takes a proof search tree $D$ as input and returns either

(a) the message $\mathfrak{t}(D)$ indicating that 
no leaf hypersequent of $D$ contains any logical symbol, or

(b) a non-atomic $\RPLA$-formula occurrence $\mathfrak{t}(D)$
(as a sequent member) in a leaf hypersequent of $D$.

By a result of a backward rule application to a proof search tree $D$ 
\emph{according to a tactic} $\mathfrak{t}$, we mean
$D$ if $\mathfrak{t}(D)$ is not a formula occurrence;
otherwise, a proof search tree obtained from $D$ by a backward application of
a (uniquely determined) rule of $\GthreeLukA$ to the occurrence $\mathfrak{t}(D)$.
We say that a proof search tree (in particular, a proof) $D$ for $\mathcal{H}$ 
can be constructed \emph{according to a tactic} $\mathfrak{t}$ 
if $D$ can be obtained from $\mathcal{H}$ by a finite number of 
backward rule applications according to $\mathfrak{t}$.

For a tactic $\mathfrak{t}$ and a hypersequent $\mathcal{H}$,
let $\mathcal{D}^{\mathfrak{t}}_\mathcal{H}$ be a tree obtained 
from $\mathcal{H}$ by an infinite number of backward applications
according to $\mathfrak{t}$.
Call a tactic $\mathfrak{t}$ \emph{fair} if,
for each hypersequent $\mathcal{H}$, 
each branch $\mathcal{B}$ of the tree $\mathcal{D}^{\mathfrak{t}}_\mathcal{H}$,
and each non-atomic $\RPLA$-formula occurrence $\mathcal{F}$ 
(as a sequent member) on $\mathcal{B}$, 
there is a backward application to some ancestor of $\mathcal{F}$ 
on $\mathcal{B}$.

Now we state a theorem that allows us to justify the use of any fair tactic
for bottom-up proof search.

\begin{theorem}  \label{GTh:RearrangingGthreeProofAccordTactics}
Suppose $\mathcal{G}$ is a $\GthreeLukA$-provable hypersequent, and
$\mathfrak{t}$ is a fair tactic.
Then some $\GthreeLukA$-proof of $\mathcal{G}$ can be constructed
according to $\mathfrak{t}$.
\end{theorem}

Before proving this theorem, we establish the following lemma, which
helps us to make one step in transforming a $\GthreeLukA$-proof
according to a tactic.

\begin{lemma}  \label{GLem:StepRearrangingGthreeProof}
Suppose $D$ is a $\GthreeLukA$-proof for $\mathcal{H}$, and
$\mathcal{F}$ is a non-atomic $\RPLA$-for\-mula occurrence 
\textnormal{(}as a sequent member\textnormal{)} in $\mathcal{H}$.
Then a $\GthreeLukA$-proof $\widehat{D}$ of the form

\medskip
\begin{center}
\def\ScoreOverhang{0pt}
\AxiomC{$\widehat{D}_1$} \noLine
\UnaryInfC{$\widehat{\mathcal{H}}_1$} 
\RightLabel{\qquad or \quad}
\UnaryInfC{$\mathcal{H}$}
\DisplayProof
\def\ScoreOverhang{0pt}
\AxiomC{$\widehat{D}_1$} \noLine
\UnaryInfC{$\widehat{\mathcal{H}}_1$}
\AxiomC{$\widehat{D}_2$} \noLine
\UnaryInfC{$\widehat{\mathcal{H}}_2$}
\BinaryInfC{$\mathcal{H}$}
\DisplayProof
\end{center}         
\medskip
can be constructed such that:

$(1)$ 
$\mathcal{F}$ is the principal formula occurrence 
in the lowest backward application in $\widehat{D}$, 
and \,${h(\widehat{D}_i) \leqslant h(D)}$ for each $i$;

$(2)$ 
if $\mathcal{F}$ is the principal formula occurrence 
in the lowest backward application in $D$, 
then $\widehat{D}$ is the same as $D$; 

$(3)$ 
if ${h(D)>0}$ and 
the principal formula occurrence $\mathcal{F}_0$ 
in the lowest backward application in $D$ differs from $\mathcal{F}$, then,
for each $i$, the ancestor of $\mathcal{F}_0$ in $\widehat{\mathcal{H}}_i$ 
is the princilal formula occurrence 
in the lowest backward application in $\widehat{D}_i$.\footnote{
  Roughly speaking, the lowest backward application in $D$
  goes one level up in $\widehat{D}$.
}
\end{lemma}

\begin{proof}
If $\mathcal{F}$ is the principal formula occurrence 
in the lowest backward application in $D$, then
we immediately take $D$ as $\widehat{D}$, and 
assertions (1)--(3) of the lemma clearly hold.

Suppose $\mathcal{F}$ is not the principal formula occurrence 
in the lowest backward application in $D$.
Then assertion (2) of the lemma is trivially true. 
Let $\mathcal{R}$ be the only inference rule that can be applied backward
to the occurrence $\mathcal{F}$ in $\mathcal{H}$.

Using the construction in the proof of the hp-invertibility of $\mathcal{R}$ 
(see Lemma \ref{GLem:InvGthreeRules}), from 
the proof $D$ for $\mathcal{H}$,
we construct proofs $\widehat{D}_i$ (${i=1}$ or ${i=1,2}$) 
for all the premises of a backward application of $\mathcal{R}$
to the occurrence $\mathcal{F}$ in $\mathcal{H}$, 
and we have ${h(\widehat{D}_i) \leqslant h(D)}$.

Now, by applying $\mathcal{R}$ to the root hypersequents of
the proofs $\widehat{D}_i$, we obtain a proof $\widehat{D}$ of $\mathcal{H}$
for which assertion (1) of the lemma holds.

After examining the construction in the proof of the hp-invertibility of
$\mathcal{R}$ (see Lemma \ref{GLem:InvGthreeRules}), we are sure that 
$\widehat{D}$ satisfies assertion (3) of the lemma being proved.
\end{proof}

\emph{Proof} of Theorem~\ref{GTh:RearrangingGthreeProofAccordTactics}.
Fix a $\GthreeLukA$-proof $D_0$ for $\mathcal{G}$ and transform 
it according to $\mathfrak{t}$ in stages.
The result of each stage will be some $\GthreeLukA$-proof $D$ 
for $\mathcal{G}$ consisting of

(a) a proof search tree $D^\mathfrak{t}$ that 
has the common root with $D$
and is constructed according to $\mathfrak{t}$,
and which is called the \emph{transformed part} of $D$, as well as

(b) a finite number of proof trees whose roots are leaves of $D^\mathfrak{t}$,
and each of which is called a \emph{nontransformed part} of $D$.

Define the transformed part of the initial proof $D_0$ to be its root, and
the only nontransformed part of it to be $D_0$ itself.

We use induction on the maximal height $H(D)$
of the nontransformed parts of the current proof $D$ being transformed.

1. If ${H(D) = 0}$, then $D$ is the required proof. 

2. Suppose ${H(D) > 0}$ and $D^\mathfrak{t}$ is the transformed part of $D$.

2.1. To obtain a proof $\widetilde{D}$ 
(with its transformed part $\widetilde{D}^\mathfrak{t}$) 
as a result of the stage, we carry out 
some finite number $N$ of backward applications to the transformed part
of the current proof (which is $D$ initially) 
according to the fair tactic $\mathfrak{t}$.
We choose such a number $N$ so that, 
for each branch $\mathcal{B}$ of $\widetilde{D}^\mathfrak{t}$
and each non-atomic $\RPLA$-formula occurrence $\mathcal{F}$
(as a sequent member) 
in the node of $\widetilde{D}^\mathfrak{t}$ that
was a leaf of $D^\mathfrak{t}$ and is on $\mathcal{B}$ now, 
there is a backward application to some ancestor of $\mathcal{F}$ 
on $\mathcal{B}$.

2.2.
We carry out each backward application 
to a formula occurrence $\mathcal{F}$ (chosen by $\mathfrak{t}$)
in a leaf of the transformed part $\mathcal{D}^\mathfrak{t}$ 
of the current proof $\mathcal{D}$ for $\mathcal{G}$ as follows. 
Let $\mathcal{D}^\mathfrak{n}$ be 
the nontransformed part of $\mathcal{D}$ whose root is this leaf, 
and $\mathcal{H}$ be the root hypersequent of $\mathcal{D}^\mathfrak{n}$.
By Lemma \ref{GLem:StepRearrangingGthreeProof},
given the proof $\mathcal{D}^\mathfrak{n}$ and the occurrence $\mathcal{F}$ 
in $\mathcal{H}$, we construct a proof $\widehat{\mathcal{D}}^\mathfrak{n}$ 
of the form

\medskip
\begin{center}
\def\ScoreOverhang{0pt}
\AxiomC{$\widehat{\mathcal{D}}^\mathfrak{n}_1$} \noLine
\UnaryInfC{$\widehat{\mathcal{H}}_1$} 
\RightLabel{\qquad or \quad}
\UnaryInfC{$\mathcal{H}$}
\DisplayProof
\def\ScoreOverhang{0pt}
\AxiomC{$\widehat{\mathcal{D}}^\mathfrak{n}_1$} \noLine
\UnaryInfC{$\widehat{\mathcal{H}}_1$}
\AxiomC{$\widehat{\mathcal{D}}^\mathfrak{n}_2$} \noLine
\UnaryInfC{$\widehat{\mathcal{H}}_2$}
\BinaryInfC{$\mathcal{H}$}
\DisplayProof
\end{center}
\medskip
such that:

$(1^{\ref{GLem:StepRearrangingGthreeProof}})$ 
$\mathcal{F}$ is the principal formula occurrence 
in the lowest backward application in $\widehat{\mathcal{D}}^\mathfrak{n}$, and 
\,${h(\widehat{\mathcal{D}}^\mathfrak{n}_i) \leqslant h(\mathcal{D}^\mathfrak{n})}$
for each $i$;

$(2^{\ref{GLem:StepRearrangingGthreeProof}})$ 
if $\mathcal{F}$ is the principal formula occurrence 
in the lowest backward application in $\mathcal{D}^\mathfrak{n}$, then 
$\widehat{\mathcal{D}}^\mathfrak{n}$ is the same as $\mathcal{D}^\mathfrak{n}$,
and hence 
${h(\widehat{\mathcal{D}}^\mathfrak{n}_i) < h(\mathcal{D}^\mathfrak{n})}$
for each $i$;

$(3^{\ref{GLem:StepRearrangingGthreeProof}})$ 
if ${h(\mathcal{D}^\mathfrak{n})>0}$ and 
the principal formula occurrence $\mathcal{F}_0$ 
in the lowest backward application in $\mathcal{D}^\mathfrak{n}$ 
differs from $\mathcal{F}$, then,
for each $i$, the ancestor of $\mathcal{F}_0$ in $\widehat{\mathcal{H}}_i$ 
is the princilal formula occurrence 
in the lowest backward application in $\widehat{\mathcal{D}}^\mathfrak{n}_i$.

Next, we replace the subtree $\mathcal{D}^\mathfrak{n}$ in $\mathcal{D}$ 
by $\widehat{\mathcal{D}}^\mathfrak{n}$.
Finally,  
the lowest backward application in $\widehat{\mathcal{D}}^\mathfrak{n}$ 
is included in the transformed part of the resulting proof for $\mathcal{G}$.
Thereby from $\mathcal{D}^\mathfrak{n}$ we obtain
one or two new nontransformed parts:
$\widehat{\mathcal{D}}^\mathfrak{n}_1$ or
$\widehat{\mathcal{D}}^\mathfrak{n}_1$ and $\widehat{\mathcal{D}}^\mathfrak{n}_2$.

2.3. Clearly, under the given transformation of $D$ into $\widetilde{D}$,
each nontransformed part $\widetilde{D}^\mathfrak{n}$ of $\widetilde{D}$ 
is obtaned from some nontransformed part $D^\mathfrak{n}$ of $D$.
If ${h(D^\mathfrak{n}) = 0}$, then it is obvious that 
${h(\widetilde{D}^\mathfrak{n}) = 0}$.

Suppose ${h(D^\mathfrak{n}) > 0}$.
By item 2.1 and assertion $(3^{\ref{GLem:StepRearrangingGthreeProof}})$, 
when we transform $D$ into $\widetilde{D}$,
we carry out so many backward applications that
the premise of assertion $(2^{\ref{GLem:StepRearrangingGthreeProof}})$ holds 
for at least one backward application performed 
in the passage from $D^\mathfrak{n}$ to $\widetilde{D}^\mathfrak{n}$. 
Therefore ${h(\widetilde{D}^\mathfrak{n}) < h(D^\mathfrak{n})}$. 

Thus ${H(\widetilde{D}) < H(D)}$.
By the induction hypothesis applied to $\widetilde{D}$, 
we construct a proof of $\mathcal{G}$ according to $\mathfrak{t}$.
\hfill$\Box$

\section{The mid-hypersequent theorem for $\GthreeLukA$ and its consequences}
\label{Sec:MidHS}

We say that a $\GthreeLukA$-proof is a \emph{mid-hypersequent} proof
if in it all applications of propositional rules are above 
all applications of quantifier rules.

To transform some $\GthreeLukA$-proofs into mid-hypersequent ones, 
we will use the following properties (P1--P4), which express
permutability of adjacent rule applications.
In each of these properties, 
the resulting proof is displayed after the inital one.
From now on, 
if a formula (or sequent) occurrence in the conclusion of a rule application
is in boldface, 
then the occurrence is the principal one in the application.
Properties P1--P4 can be verified in a straightforward way.

\smallskip\textbf{P1.} 
Let $\mathcal{R}_1$ and $\mathcal{R}_2$ be any one-premise inference rules of
$\GthreeLukA$, except the case where
\,${\mathcal{R}_1 \in \{\, (\Rightarrow\forall)^3,\ (\exists\Rightarrow)^3 \,\}}$\,
and
\,${\mathcal{R}_2 \in \{\, (\Rightarrow\exists)^3,\ (\forall\Rightarrow)^3 \,\}}$.

If $\mathcal{R}_1$ is ${(\to\Rightarrow)^3}$,
${(\forall\Rightarrow)^3}$, or ${(\exists\Rightarrow)^3}$, 
\,$\mathcal{R}_2$ is ${(\Rightarrow\forall)^3}$ or ${(\Rightarrow\exists)^3}$, 
and the above case is excluded,
then we can perform the following transformation:
\medskip
\begin{center}
\def\ScoreOverhang{0pt}
\AxiomC{$D$} \noLine
\UnaryInfC{$\mathcal{G} \,|\, \Gamma, F_1 \Rightarrow F_2, \Delta \,|\, \mathcal{H}_1 \,|\, \mathcal{H}_2$} 
\RightLabel{$\mathcal{R}_2$}
\UnaryInfC{$\mathcal{G} \,|\, \Gamma, F_1 \Rightarrow \boldsymbol{A_2}, \Delta \,|\, \mathcal{H}_1$}
\RightLabel{$\mathcal{R}_1$}
\UnaryInfC{$\mathcal{G} \,|\, \Gamma, \boldsymbol{A_1} \Rightarrow A_2, \Delta$}
\DisplayProof
\qquad 
\def\ScoreOverhang{0pt}
\AxiomC{$D$} \noLine
\UnaryInfC{$\mathcal{G} \,|\, \Gamma, F_1 \Rightarrow F_2, \Delta \,|\, \mathcal{H}_1 \,|\, \mathcal{H}_2$}
\RightLabel{$\mathcal{R}_1$}
\UnaryInfC{$\mathcal{G} \,|\, \Gamma, \boldsymbol{A_1} \Rightarrow F_2, \Delta \,|\, \mathcal{H}_2$}
\RightLabel{$\mathcal{R}_2$}
\UnaryInfC{$\mathcal{G} \,|\, \Gamma, A_1 \Rightarrow \boldsymbol{A_2}, \Delta$}
\DisplayProof
\end{center}
\medskip

For a hypersequent that is at the bottom of an appropriate initial proof 
and has the form 
\begin{center}
${\mathcal{G} \,|\, \Gamma, A_1, A_2 \Rightarrow \Delta}$, \quad
${\mathcal{G} \,|\, \Gamma \Rightarrow A_1, A_2, \Delta}$, \quad or \quad
${\mathcal{G} \,|\, \Gamma, A_2 \Rightarrow A_1, \Delta}$,
\end{center}
we can carry out a transformation similar to that just given. 

E.g., if $\mathcal{R}_1$ is ${(\to\Rightarrow)^3}$ and
$\mathcal{R}_2$ is ${(\Rightarrow\forall)^3}$, then
the initial and resulting proofs look like:
\medskip
\begin{center} 
\def\ScoreOverhang{0pt}
\AxiomC{$D$} \noLine
\UnaryInfC{$\mathcal{G} \,|\, \Gamma, \SpV{p} \Rightarrow \RePl{C}{x}{a}, \Delta \,|\, B \Rightarrow \SpV{p}, A$} 
\RightLabel{${(\Rightarrow\forall)^3}$}
\UnaryInfC{$\mathcal{G} \,|\, \Gamma, \SpV{p} \Rightarrow \boldsymbol{\forall x C}, \Delta \,|\, B \Rightarrow \SpV{p}, A$}
\RightLabel{${(\to\Rightarrow)^3}$}
\UnaryInfC{$\mathcal{G} \,|\, \Gamma, \boldsymbol{A\to B} \Rightarrow \forall x C, \Delta$}
\DisplayProof
\qquad 
\def\ScoreOverhang{0pt}
\AxiomC{$D$} \noLine
\UnaryInfC{$\mathcal{G} \,|\, \Gamma, \SpV{p} \Rightarrow \RePl{C}{x}{a}, \Delta \,|\, B \Rightarrow \SpV{p}, A$}
\RightLabel{${(\to\Rightarrow)^3}$}
\UnaryInfC{$\mathcal{G} \,|\, \Gamma, \boldsymbol{A\to B} \Rightarrow \RePl{C}{x}{a}, \Delta$}
\RightLabel{${(\Rightarrow\forall)^3}$}
\UnaryInfC{$\mathcal{G} \,|\, \Gamma, A\to B \Rightarrow \boldsymbol{\forall x C}, \Delta$}
\DisplayProof
\end{center}
\medskip

\smallskip\textbf{P2.} 
Let rules $\mathcal{R}_1$ and $\mathcal{R}_2$ be as in the first paragraph of P1.
Then we can perform this transformation:
\medskip
\begin{center}
\def\ScoreOverhang{0pt}
\AxiomC{$D$} \noLine
\UnaryInfC{$\mathcal{G}\,|\,\mathcal{H}_1\,|\,\mathcal{H}_2$} 
\RightLabel{$\mathcal{R}_2$}
\UnaryInfC{$\mathcal{G}\,|\,\mathcal{H}_1\,|\,\boldsymbol{S_2}$} 
\RightLabel{$\mathcal{R}_1$}
\UnaryInfC{$\mathcal{G}\,|\,\boldsymbol{S_1}\,|\,S_2$}
\DisplayProof
\qquad 
\def\ScoreOverhang{0pt}
\AxiomC{$D$} \noLine
\UnaryInfC{$\mathcal{G}\,|\,\mathcal{H}_1\,|\,\mathcal{H}_2$} 
\RightLabel{$\mathcal{R}_1$}
\UnaryInfC{$\mathcal{G}\,|\,\boldsymbol{S_1}\,|\,\mathcal{H}_2$} 
\RightLabel{$\mathcal{R}_2$}
\UnaryInfC{$\mathcal{G}\,|\,S_1\,|\,\boldsymbol{S_2}$}
\DisplayProof
\end{center}
\medskip

E.g., if $\mathcal{R}_1$ is ${(\to\Rightarrow)^3}$ and
$\mathcal{R}_2$ is ${(\Rightarrow\forall)^3}$, then
the initial and resulting proofs have the forms:
\medskip
\begin{center} 
\def\ScoreOverhang{0pt}
\AxiomC{$D$} \noLine
\UnaryInfC{$\mathcal{G} \,|\, \Gamma_1, \SpV{p} \Rightarrow \Delta_1 \,|\, B \Rightarrow \SpV{p}, A \,|\, \Gamma_2 \Rightarrow \RePl{C}{x}{a}, \Delta_2$}
\RightLabel{${(\Rightarrow\forall)^3}$}
\UnaryInfC{$\mathcal{G} \,|\, \Gamma_1, \SpV{p} \Rightarrow \Delta_1 \,|\, B \Rightarrow \SpV{p}, A \,|\, \Gamma_2 \Rightarrow \boldsymbol{\forall x C}, \Delta_2$}
\RightLabel{${(\to\Rightarrow)^3}$}
\UnaryInfC{$\mathcal{G} \,|\, \Gamma_1, \boldsymbol{A\to B} \Rightarrow \Delta_1 \,|\, \Gamma_2 \Rightarrow \forall x C, \Delta_2$}
\DisplayProof
\end{center}
\medskip
\begin{center} 
\def\ScoreOverhang{0pt}
\AxiomC{$D$} \noLine
\UnaryInfC{$\mathcal{G} \,|\, \Gamma_1, \SpV{p} \Rightarrow \Delta_1 \,|\, B \Rightarrow \SpV{p}, A \,|\, \Gamma_2 \Rightarrow \RePl{C}{x}{a}, \Delta_2$} 
\RightLabel{${(\to\Rightarrow)^3}$}
\UnaryInfC{$\mathcal{G} \,|\, \Gamma_1, \boldsymbol{A\to B} \Rightarrow \Delta_1 \,|\, \Gamma_2 \Rightarrow \RePl{C}{x}{a}, \Delta_2$}
\RightLabel{${(\Rightarrow\forall)^3}$}
\UnaryInfC{$\mathcal{G} \,|\, \Gamma_1, A\to B \Rightarrow \Delta_1 \,|\, \Gamma_2 \Rightarrow \boldsymbol{\forall x C}, \Delta_2$}
\DisplayProof
\end{center}
\medskip

\smallskip\textbf{P3.} 
Let $\mathcal{R}$ be any one-premise inference rule of $\GthreeLukA$.

If $\mathcal{R}$ is ${(\to\Rightarrow)^3}$, ${(\forall\Rightarrow)^3}$, or
${(\exists\Rightarrow)^3}$, then under the conditions stated below,
we can carry out the following transformation:
\medskip
\begin{center}
\def\ScoreOverhang{0pt}
\AxiomC{$D_1$} \noLine
\UnaryInfC{$\mathcal{G} \,|\, \Gamma, F_1 \Rightarrow \Delta \,|\, \mathcal{H}_1$}
\LeftLabel{$\mathcal{R}$}
\UnaryInfC{$\mathcal{G} \,|\, \Gamma, \boldsymbol{C} \Rightarrow \Delta$}
\AxiomC{$D_2$} \noLine
\UnaryInfC{$\mathcal{G} \,|\, \Gamma, F_2, A \Rightarrow B, \Delta \,|\, \mathcal{H}_2$}
\RightLabel{$\mathcal{R}$}
\UnaryInfC{$\mathcal{G} \,|\, \Gamma, \boldsymbol{C}, A \Rightarrow B, \Delta$}
\RightLabel{$(\Rightarrow\to)^3$}
\BinaryInfC{$\mathcal{G} \,|\, \Gamma, C \Rightarrow \boldsymbol{A\to B}, \Delta$}
\DisplayProof
\end{center}
\medskip
\begin{center}
\def\ScoreOverhang{0pt}
\AxiomC{$D_1'$} \noLine
\UnaryInfC{$\mathcal{G} \,|\, \Gamma, F_2 \Rightarrow \Delta \,|\, \mathcal{H}_2$}
\AxiomC{$D_2$} \noLine
\UnaryInfC{$\mathcal{G} \,|\, \Gamma, F_2, A \Rightarrow B, \Delta \,|\, \mathcal{H}_2$}
\RightLabel{$(\Rightarrow\to)^3$}
\BinaryInfC{$\mathcal{G} \,|\, \Gamma, F_2 \Rightarrow \boldsymbol{A\to B}, \Delta \,|\, \mathcal{H}_2$}
\RightLabel{$\mathcal{R}$}
\UnaryInfC{$\mathcal{G} \,|\, \Gamma, \boldsymbol{C} \Rightarrow A\to B, \Delta$}
\DisplayProof
\end{center}
\medskip

If $\mathcal{R}$ is ${(\Rightarrow\forall)^3}$ or ${(\Rightarrow\exists)^3}$, 
then under the conditions stated below,
we can perform a similar transformation with a bottom hypersequent of the form
\,$\mathcal{G} \,|\, \Gamma \Rightarrow C, A\to B, \Delta$.

For both the transformations, two conditions must hold.
First, 
if $\mathcal{R}$ is ${(\forall\Rightarrow)^3}$ or ${(\Rightarrow\exists)^3}$, 
then the proper terms of the three displayed 
applications of $\mathcal{R}$ are the same.
Second, we construct the proof $D_1'$ thus:

(a) Suppose $\mathcal{R}$ is
${(\exists\Rightarrow)^3}$ or ${(\Rightarrow\forall)^3}$,
\,$a_1$ and $a_2$ are the proper parameters of the two applications 
of $\mathcal{R}$ displayed in the initial proof
on the left and right, respectively;
then: ${D_1' = D_1}$ if ${a_1 = a_2}$; otherwise,
we obtain the proof $\widetilde{D}_1$ (for the root hypersequent of $D_1$)
from $D_1$ by replacing all occurrences of $a_2$ with 
a parameter not occurring in $D_1$, and
next, we get the required proof $D_1'$ from $\widetilde{D}_1$ 
by replacing all occurrences of $a_1$ with $a_2$.

(b) If $\mathcal{R}$ is ${(\to\Rightarrow)^3}$, ${(\forall\Rightarrow)^3}$, or
${(\Rightarrow\exists)^3}$, then
we obtain $D_1'$ from $D_1$ as in (a), but instead of parameters, 
we use semipropositional variables of the type corresponding to the rule
$\mathcal{R}$.

E.g., if $\mathcal{R}$ is ${(\forall\Rightarrow)^3}$, 
then the initial and resulting proofs look like:
\medskip
\begin{center} 
\def\ScoreOverhang{0pt}
\AxiomC{$D_1$} \noLine
\UnaryInfC{$\genfrac{}{}{0pt}{}{\displaystyle \mathcal{G} \,|\, \Gamma, \SpV{p}_1 \Rightarrow \Delta}
            {\displaystyle |\, \forall x C \Rightarrow \SpV{p}_1 \,|\, \RePl{C}{x}{t} \Rightarrow \SpV{p}_1}$}
\LeftLabel{${(\forall\Rightarrow)^3}$}
\UnaryInfC{$\mathcal{G} \,|\, \Gamma, \boldsymbol{\forall x C} \Rightarrow \Delta$}
\AxiomC{$D_2$} \noLine
\UnaryInfC{$\genfrac{}{}{0pt}{}{\displaystyle \mathcal{G} \,|\, \Gamma, \SpV{p}_2, A \Rightarrow B, \Delta}
            {\displaystyle |\, \forall x C \Rightarrow \SpV{p}_2 \,|\, \RePl{C}{x}{t} \Rightarrow \SpV{p}_2}$}
\RightLabel{${(\forall\Rightarrow)^3}$}
\UnaryInfC{$\mathcal{G} \,|\, \Gamma, \boldsymbol{\forall x C}, A \Rightarrow B, \Delta$}
\RightLabel{$(\Rightarrow\to)^3$}
\BinaryInfC{$\mathcal{G} \,|\, \Gamma, \forall x C \Rightarrow \boldsymbol{A\to B}, \Delta$}
\DisplayProof
\end{center}
\medskip
\begin{center} 
\def\ScoreOverhang{0pt}
\AxiomC{$D_1'$} \noLine
\UnaryInfC{$\genfrac{}{}{0pt}{}{\displaystyle \mathcal{G} \,|\, \Gamma, \SpV{p}_2 \Rightarrow \Delta}
            {\displaystyle |\, \forall x C \Rightarrow \SpV{p}_2 \,|\, \RePl{C}{x}{t} \Rightarrow \SpV{p}_2}$}
\AxiomC{$D_2$} \noLine
\UnaryInfC{$\genfrac{}{}{0pt}{}{\displaystyle \mathcal{G} \,|\, \Gamma, \SpV{p}_2, A \Rightarrow B, \Delta}
            {\displaystyle |\, \forall x C \Rightarrow \SpV{p}_2 \,|\, \RePl{C}{x}{t} \Rightarrow \SpV{p}_2}$}
\RightLabel{$(\Rightarrow\to)^3$}
\BinaryInfC{$\mathcal{G} \,|\, \Gamma, \SpV{p}_2 \Rightarrow \boldsymbol{A\to B}, \Delta \,|\, \forall x C \Rightarrow \SpV{p}_2 \,|\, \RePl{C}{x}{t} \Rightarrow \SpV{p}_2$}
\RightLabel{${(\forall\Rightarrow)^3}$}
\UnaryInfC{$\mathcal{G} \,|\, \Gamma, \boldsymbol{\forall x C} \Rightarrow A\to B, \Delta$}
\DisplayProof
\end{center}
\medskip

\smallskip\textbf{P4.} 
Let $\mathcal{R}$ be any one-premise inference rule of $\GthreeLukA$.
Then we can carry out the transformation:
\medskip
\begin{center}
\def\ScoreOverhang{0pt}
\AxiomC{$D_1$} \noLine
\UnaryInfC{$\mathcal{G} \,|\, \mathcal{H}_{1,1} \,|\, \mathcal{H}_{2,0}$} 
\LeftLabel{$\mathcal{R}$}
\UnaryInfC{$\mathcal{G} \,|\, \mathcal{H}_{1,1} \,|\, \boldsymbol{S_2}$} 
\AxiomC{$D_2$} \noLine
\UnaryInfC{$\mathcal{G} \,|\, \mathcal{H}_{1,2} \,|\, \mathcal{H}_2$} 
\RightLabel{$\mathcal{R}$}
\UnaryInfC{$\mathcal{G} \,|\, \mathcal{H}_{1,2} \,|\, \boldsymbol{S_2}$} 
\RightLabel{$(\Rightarrow\to)^3$}
\BinaryInfC{$\mathcal{G} \,|\, \boldsymbol{S_1} \,|\, S_2$}
\DisplayProof
\quad 
\def\ScoreOverhang{0pt}
\AxiomC{$D_1'$} \noLine
\UnaryInfC{$\mathcal{G} \,|\, \mathcal{H}_{1,1} \,|\, \mathcal{H}_2$}
\AxiomC{$D_2$} \noLine
\UnaryInfC{$\mathcal{G} \,|\, \mathcal{H}_{1,2} \,|\, \mathcal{H}_2$}
\RightLabel{$(\Rightarrow\to)^3$}
\BinaryInfC{$\mathcal{G} \,|\, \boldsymbol{S_1} \,|\, \mathcal{H}_2$} 
\RightLabel{$\mathcal{R}$}
\UnaryInfC{$\mathcal{G} \,|\, S_1 \,|\, \boldsymbol{S_2}$}
\DisplayProof
\end{center}
\medskip
Here all the principal formula occurrences in the three displayed 
applications of $\mathcal{R}$ represent the same formula;
in the case 
where $\mathcal{R}$ is ${(\forall\Rightarrow)^3}$ or ${(\Rightarrow\exists)^3}$,
the additional condition is the same as in P3;
and the proof $D_1'$ is constructed from $D_1$ as in P3.

E.g., if $\mathcal{R}$ is ${(\Rightarrow\forall)^3}$, then
the initial and resulting proofs have the forms:
\medskip
\begin{center} 
\def\ScoreOverhang{0pt}
\AxiomC{$D_1$} \noLine
\UnaryInfC{$\mathcal{G} \,|\, \Gamma_1 \Rightarrow \Delta_1 \,|\, \Gamma_2 \Rightarrow \RePl{C}{x}{a_1}, \Delta_2$} 
\LeftLabel{${(\Rightarrow\forall)^3}$}
\UnaryInfC{$\mathcal{G} \,|\, \Gamma_1 \Rightarrow \Delta_1 \,|\, \Gamma_2 \Rightarrow \boldsymbol{\forall x C}, \Delta_2$} 
\AxiomC{$D_2$} \noLine
\UnaryInfC{$\mathcal{G} \,|\, \Gamma_1,A \Rightarrow B,\Delta_1 \,|\, \Gamma_2 \Rightarrow \RePl{C}{x}{a_2}, \Delta_2$}
\RightLabel{${(\Rightarrow\forall)^3}$}
\UnaryInfC{$\mathcal{G} \,|\, \Gamma_1,A \Rightarrow B,\Delta_1 \,|\, \Gamma_2 \Rightarrow \boldsymbol{\forall x C}, \Delta_2$} 
\RightLabel{$(\Rightarrow\to)^3$}
\BinaryInfC{$\mathcal{G} \,|\, \Gamma_1 \Rightarrow \boldsymbol{A\to B}, \Delta_1 \,|\, \Gamma_2 \Rightarrow \forall x C, \Delta_2$}
\DisplayProof
\end{center}
\medskip
\begin{center} 
\def\ScoreOverhang{0pt}
\AxiomC{$D_1'$} \noLine
\UnaryInfC{$\mathcal{G} \,|\, \Gamma_1 \Rightarrow \Delta_1 \,|\, \Gamma_2 \Rightarrow \RePl{C}{x}{a_2}, \Delta_2$}
\AxiomC{$D_2$} \noLine
\UnaryInfC{$\mathcal{G} \,|\, \Gamma_1,A \Rightarrow B,\Delta_1 \,|\, \Gamma_2 \Rightarrow \RePl{C}{x}{a_2}, \Delta_2$}
\RightLabel{$(\Rightarrow\to)^3$}
\BinaryInfC{$\mathcal{G} \,|\, \Gamma_1 \Rightarrow \boldsymbol{A\to B}, \Delta_1 \,|\, \Gamma_2 \Rightarrow \RePl{C}{x}{a_2}, \Delta_2$} 
\RightLabel{${(\Rightarrow\forall)^3}$}
\UnaryInfC{$\mathcal{G} \,|\, \Gamma_1 \Rightarrow A\to B, \Delta_1 \,|\, \Gamma_2 \Rightarrow \boldsymbol{\forall x C}, \Delta_2$}
\DisplayProof
\end{center}
\smallskip

\begin{theorem}[the mid-hypersequent theorem for $\GthreeLukA$]
\label{GTh:MidHsGthree}
Let $\mathcal{H}$ be a hypersequent in which 
each member of each sequent is a prenex $\RPLA$-formula or 
a semipropositional variable.
Then any $\GthreeLukA$-proof $D$ for $\mathcal{H}$ 
can be transformed into a mid-hyper\-sequent $\GthreeLukA$-proof $\widehat{D}$ 
for $\mathcal{H}$;
moreover, ${Q(\widehat{D}) \leqslant Q(D)}$, where $Q(\mathcal{D})$ is
the number of quantifier rule applications in 
a $\GthreeLukA$-proof $\mathcal{D}$. 
\end{theorem}

\begin{proof}
For a $\GthreeLukA$-proof $\mathcal{D}$ and a propositional rule application $R$ 
in $\mathcal{D}$, let $\mathcal{O}(R)$ be
the number of quantifier rule applications above $R$, 
and $\mathcal{O}(\mathcal{D})$ be the sum of $\mathcal{O}(R)$ 
over all propositional rule applications $R$ in $\mathcal{D}$.

We proceed by induction on $\mathcal{O}(D)$, 
where $D$ is a given proof for $\mathcal{H}$.

1. If ${\mathcal{O}(D) = 0}$, then $D$ is the desired proof.

2. Otherwise, 
choose an application $R_0$ of a (propositional) rule $\mathcal{R}_0$ in $D$
such that ${\mathcal{O}(R_0) > 0}$ and 
no application $R'$ with ${\mathcal{O}(R') > 0}$ is above $R_0$.

2.1. Suppose $\mathcal{R}_0$ is ${(\to\Rightarrow)^3}$.
By $R_1$ denote the (quantifier) rule application
that stands immediately above the application $R_0$.
We permute $R_0$ and $R_1$ using transformation P1 or P2, 
and next, by the induction hypothesis, we obtain the desired proof.

2.2. Now suppose $\mathcal{R}_0$ is ${(\Rightarrow\to)^3}$, and 
the proof for the conclusion $\mathcal{H}_0$ of the application $R_0$ 
looks like:
\begin{center}
\def\ScoreOverhang{0pt}
\AxiomC{$D_1$} \noLine
\UnaryInfC{$\mathcal{H}_1$}
\AxiomC{$D_2$} \noLine
\UnaryInfC{$\mathcal{H}_2$}
\RightLabel{$\mathcal{R}_0$.}
\BinaryInfC{$\mathcal{H}_0$}
\DisplayProof
\end{center}
\smallskip
Then the lowest application in $D_1$ or $D_2$, 
say for definiteness the lowest application $R_2$ in $D_2$, 
is an application of a quantifier rule $\mathcal{R}$.

By the induction hypothesis, we can transform $D_1$ into 
a mid-hypersequent proof $D_1'$ for $\mathcal{H}_1$ 
such that ${Q(D_1') \leqslant Q(D_1)}$.
In the proof $D$ (for $\mathcal{H}$), we replace the subtree $D_1$ by $D_1'$,
thus obtaining a proof $D'$ for $\mathcal{H}$.

Let the principal formula occurrence $\mathcal{F}_2$ in $R_2$
(which is a formula occurrence in $\mathcal{H}_2$)
be an ancestor of an occurrence $\mathcal{F}_0$ in $\mathcal{H}_0$.
The formulas $A$ and $B$ in $\mathcal{H}_2$ that originate from 
the principal occurrence of ${A\to B}$ in $R_0$ are quantifier-free.
Therefore the occurrence $\mathcal{F}_0$ has an ancestor $\mathcal{F}_1$ 
in $\mathcal{H}_1$, and 
all $\mathcal{F}_i$ (${i=0,1,2}$) represent the same formula.

Using the construction in the proof of the hp-invertibility of 
the rule $\mathcal{R}$ (see Lemma \ref{GLem:InvGthreeRules}),
from the proof $D_1'$ for $\mathcal{H}_1$, we construct a proof $D_1''$
for the premise of an application $R_1$ of $\mathcal{R}$ 
with $\mathcal{H}_1$ as the conclusion and 
$\mathcal{F}_1$ as the principal formula occurrence. 
Here if $\mathcal{R}$ is
${(\forall\Rightarrow)^3}$ or ${(\Rightarrow\exists)^3}$, then
the proper term of the application $R_1$ (of $\mathcal{R}$) 
is taken to be the proper term of the application $R_2$ (of $\mathcal{R}$).
Let $D_1'''$ be the proof (for $\mathcal{H}_1$) obtained from the proof $D_1''$ 
by the application $R_1$. 

Given the mid-hypersequent proof $D_1'$, it is not hard to see that
$D_1''$ is also a mid-hypersequent proof
$\bigl($i.e., ${\mathcal{O}(D_1'') = 0}${}$\bigr)$ and
${Q(D_1'') \leqslant Q(D_1')}$. 
Then obviously, 
${\mathcal{O}(D_1''') = 0}$ and ${Q(D_1''') \leqslant Q(D_1')+1}$. 

Next, in the proof $D'$ (for $\mathcal{H}$), 
we replace the subtree $D_1'$ by $D_1'''$ 
and get a proof $D''$ for $\mathcal{H}$.
Using transformation P3 or P4, in $D''$ we permute 
the application $R_0$ (of the two-premise rule $\mathcal{R}_0$) and
the applications $R_1$ and $R_2$ (of the quantifier rule $\mathcal{R}$),
which stand immediately above $R_0$;
and we have a proof $D'''$ for $\mathcal{H}$ as a result.

From ${\mathcal{O}(D_1''') = 0}$,
\,${Q(D_1''') \leqslant Q(D_1')+1 \leqslant Q(D_1)+1}$, and
the forms of transformations P3 and P4,
it follows that 
${\mathcal{O}(D''') < \mathcal{O}(D)}$ and ${Q(D''') \leqslant Q(D)}$.
Then by the induction hypothesis, we can construct the desired proof from $D'''$.
\end{proof}

\noindent\textbf{Remark 2.}
In contrast to Theorem \ref{GTh:MidHsGthree} above,
Theorems 10 and 18 in \cite{Ger2017} 
(i.e., the mid-hypersequent theorems for $\GoneLukA$ and $\GtwoLukA$) 
require an initial hypersequent to be of the form 
\,${\Rightarrow A}$,\, where $A$ is a prenex $\RPLA$-formula.

\begin{theorem} \label{GTh:GoneGtwoGthreeEquivForPrenexFormulas}
Suppose $A$ is a prenex $\RPLA$-formula.
Then the following are equivalent: 
\textnormal{(1)}~${\vdash_{\GoneLukA} A}$,
\,\textnormal{(2)}~${\vdash_{\GtwoLukA} A}$, 
\,\textnormal{(3)}~${\vdash_{\GthreeLukA} A}$.
\end{theorem}

\begin{proof}
(1) and (2) are equivalent by Theorem 15 in \cite{Ger2017}. 
(3) follows from (1) by Theorem \ref{GTh:GOneProvGthreeProv}.
We will show that (3) implies (2).

In view of Theorem \ref{GTh:MidHsGthree}, it is enough to 
transform any mid-hypersequent $\GthreeLukA$-proof $D_3$ for $A$
into some $\GtwoLukA$-proof for $A$.
Let $\widetilde{D}_3$ be the $\GthreeLukA$-proof search tree for $A$
consisting of all quantifier rule applications in $D_3$;\,
$\mathcal{H}_3$ be the only top hypersequent in $\widetilde{D}_3$;
and $\widetilde{\mathcal{H}}_3$ be the hypersequent that is obtained 
from $\mathcal{H}_3$ by removing all sequents containing quantifiers.

To avoid cumbersome notation,
first we will perform the transformation in the case when
$A$ has the form $\exists x \forall y \exists z B(x,y,z)$
(where $B(x,y,z)$ is a quantfier-free $\RPLA$-formula, 
and $x,y,z$ are distinct variables), and 
$\widetilde{D}_3$ has the form given in Figure \ref{fig:GthreeProofSearchTree};
then we will explain why a similar transformation can be carried out 
in the general case.
The result of simultaneously replacing all occurrences of $x,y,z$ in $B(x,y,z)$ 
with terms $s_1,s_2,s_3$, respectively, is denoted by $B(s_1,s_2,s_3)$. 

\begin{figure}[!t]
\begin{center} \small
\def\ScoreOverhang{0pt}
\AxiomC{$\genfrac{}{}{0pt}{}{
 \genfrac{}{}{0pt}{}{\displaystyle \Rightarrow \SpV{q}_1 \,|\, \SpV{q}_1 \Rightarrow \SpV{q}_3 \,|\, \SpV{q}_3 \Rightarrow \exists x \forall y \exists z B(x,y,z) \,|\, \SpV{q}_3 \Rightarrow \SpV{q}_4} 
                    {\displaystyle |\, \SpV{q}_4 \Rightarrow \SpV{q}_5 \,|\, \SpV{q}_5 \Rightarrow \exists z B(t_3,a_2,z) \,|\, \SpV{q}_5 \Rightarrow B(t_3,a_2,t_5)}
 }{\displaystyle |\, \SpV{q}_4 \Rightarrow B(t_3,a_2,t_4) \,|\, \SpV{q}_1 \Rightarrow \SpV{q}_2 \,|\, \SpV{q}_2 \Rightarrow \exists z B(t_1,a_1,z) \,|\, \SpV{q}_2 \Rightarrow B(t_1,a_1,t_2)}$}
\RightLabel{$(\Rightarrow\exists)^3$}
\UnaryInfC{$\genfrac{}{}{0pt}{}{\displaystyle \Rightarrow \SpV{q}_1 \,|\, \SpV{q}_1 \Rightarrow \SpV{q}_3 \,|\, \SpV{q}_3 \Rightarrow \exists x \forall y \exists z B(x,y,z) \,|\, \SpV{q}_3 \Rightarrow \SpV{q}_4 \,|\, \SpV{q}_4 \Rightarrow \boldsymbol{\exists z B(t_3,a_2,z)}}{\displaystyle |\, \SpV{q}_4 \Rightarrow B(t_3,a_2,t_4) \,|\, \SpV{q}_1 \Rightarrow \SpV{q}_2 \,|\, \SpV{q}_2 \Rightarrow \exists z B(t_1,a_1,z) \,|\, \SpV{q}_2 \Rightarrow B(t_1,a_1,t_2)}$}
\RightLabel{$(\Rightarrow\exists)^3$}
\UnaryInfC{$\genfrac{}{}{0pt}{}{\displaystyle \Rightarrow \SpV{q}_1 \,|\, \SpV{q}_1 \Rightarrow \SpV{q}_3 \,|\, \SpV{q}_3 \Rightarrow \exists x \forall y \exists z B(x,y,z) \,|\, \SpV{q}_3 \Rightarrow \boldsymbol{\exists z B(t_3,a_2,z)}}{\displaystyle |\, \SpV{q}_1 \Rightarrow \SpV{q}_2 \,|\, \SpV{q}_2 \Rightarrow \exists z B(t_1,a_1,z) \,|\, \SpV{q}_2 \Rightarrow B(t_1,a_1,t_2)}$}
\RightLabel{$(\Rightarrow\forall)^3$}
\UnaryInfC{$\genfrac{}{}{0pt}{}{\displaystyle \Rightarrow \SpV{q}_1 \,|\, \SpV{q}_1 \Rightarrow \SpV{q}_3 \,|\, \SpV{q}_3 \Rightarrow \exists x \forall y \exists z B(x,y,z) \,|\, \SpV{q}_3 \Rightarrow \boldsymbol{\forall y \exists z B(t_3,y,z)}}{\displaystyle |\, \SpV{q}_1 \Rightarrow \SpV{q}_2 \,|\, \SpV{q}_2 \Rightarrow \exists z B(t_1,a_1,z) \,|\, \SpV{q}_2 \Rightarrow B(t_1,a_1,t_2)}$}
\RightLabel{$(\Rightarrow\exists)^3$}
\UnaryInfC{$\genfrac{}{}{0pt}{}{\displaystyle \Rightarrow \SpV{q}_1 \,|\, \SpV{q}_1 \Rightarrow \boldsymbol{\exists x \forall y \exists z B(x,y,z)}}{\displaystyle |\, \SpV{q}_1 \Rightarrow \SpV{q}_2 \,|\, \SpV{q}_2 \Rightarrow \exists z B(t_1,a_1,z) \,|\, \SpV{q}_2 \Rightarrow B(t_1,a_1,t_2)}$}
\RightLabel{$(\Rightarrow\exists)^3$}
\UnaryInfC{$\Rightarrow \SpV{q}_1 \,|\, \SpV{q}_1 \Rightarrow \exists x \forall y \exists z B(x,y,z) \,|\, \SpV{q}_1 \Rightarrow \boldsymbol{\exists z B(t_1,a_1,z)}$}
\RightLabel{$(\Rightarrow\forall)^3$}
\UnaryInfC{$\Rightarrow \SpV{q}_1 \,|\, \SpV{q}_1 \Rightarrow \exists x \forall y \exists z B(x,y,z) \,|\, \SpV{q}_1 \Rightarrow \boldsymbol{\forall y \exists z B(t_1,y,z)}$}
\RightLabel{$(\Rightarrow\exists)^3$}
\UnaryInfC{$\Rightarrow \boldsymbol{\exists x \forall y \exists z B(x,y,z)}$}
\DisplayProof

\caption{The $\GthreeLukA$-proof search tree $\widetilde{D}_3$}
\label{fig:GthreeProofSearchTree}
\end{center}
\end{figure}
\begin{figure}[!t]
\begin{center} \small
\def\ScoreOverhang{0pt}
\AxiomC{$\genfrac{}{}{0pt}{}{\displaystyle \Rightarrow \exists x \forall y \exists z B(x,y,z) \,| \Rightarrow \exists z B(t_3,a_2,z) \,| \Rightarrow B(t_3,a_2,t_5) \,| \Rightarrow B(t_3,a_2,t_4)}{\displaystyle \,| \Rightarrow \exists z B(t_1,a_1,z) \,| \Rightarrow B(t_1,a_1,t_2)}$}
\RightLabel{$(\Rightarrow\exists)^2$}
\UnaryInfC{$\genfrac{}{}{0pt}{}{\displaystyle \Rightarrow \exists x \forall y \exists z B(x,y,z) \,| \Rightarrow \boldsymbol{\exists z B(t_3,a_2,z)} \,| \Rightarrow B(t_3,a_2,t_4)}{\displaystyle \,| \Rightarrow \exists z B(t_1,a_1,z) \,| \Rightarrow B(t_1,a_1,t_2)}$}
\RightLabel{$(\Rightarrow\exists)^2$}
\UnaryInfC{$\Rightarrow \exists x \forall y \exists z B(x,y,z) \,| \Rightarrow \boldsymbol{\exists z B(t_3,a_2,z)} \,| \Rightarrow \exists z B(t_1,a_1,z) \,| \Rightarrow B(t_1,a_1,t_2)$}
\RightLabel{$(\Rightarrow\forall)^2$}
\UnaryInfC{$\Rightarrow \exists x \forall y \exists z B(x,y,z) \,| \Rightarrow \boldsymbol{\forall y \exists z B(t_3,y,z)} \,| \Rightarrow \exists z B(t_1,a_1,z) \,| \Rightarrow B(t_1,a_1,t_2)$}
\RightLabel{$(\Rightarrow\exists)^2$}
\UnaryInfC{$\Rightarrow \boldsymbol{\exists x \forall y \exists z B(x,y,z)} \,| \Rightarrow \exists z B(t_1,a_1,z) \,| \Rightarrow B(t_1,a_1,t_2)$}
\RightLabel{$(\Rightarrow\exists)^2$}
\UnaryInfC{$\Rightarrow \exists x \forall y \exists z B(x,y,z) \,| \Rightarrow \boldsymbol{\exists z B(t_1,a_1,z)}$}
\RightLabel{$(\Rightarrow\forall)^2$}
\UnaryInfC{$\Rightarrow \exists x \forall y \exists z B(x,y,z) \,| \Rightarrow \boldsymbol{\forall y \exists z B(t_1,y,z)}$}
\RightLabel{$(\Rightarrow\exists)^2$}
\UnaryInfC{$\Rightarrow \boldsymbol{\exists x \forall y \exists z B(x,y,z)}$}
\DisplayProof

\caption{The $\GtwoLukA$-proof search tree $\widetilde{D}_2$}
\label{fig:GtwoProofSearchTree}
\end{center}
\end{figure}

From $\widetilde{D}_3$ we can construct 
the $\GtwoLukA$-proof search tree $\widetilde{D}_2$ 
given in Figure \ref{fig:GtwoProofSearchTree}
by starting with the hypersequent \,${\Rightarrow A}$\, and applying 
the rules ${(\Rightarrow\exists)^2}$ and ${(\Rightarrow\forall)^2}$ backward,
according to how
the rules ${(\Rightarrow\exists)^3}$ and ${(\Rightarrow\forall)^3}$ 
are applied backward in $\widetilde{D}_3$ 
(such a correspondence between rule applications 
is natural and is not described for brevity).
Let $\mathcal{H}_2$ be the top hypersequent in $\widetilde{D}_2$;
and $\widetilde{\mathcal{H}}_2$ be the hypersequent consisting of all
quantifier-free sequents of $\mathcal{H}_2$. 

To complete our proof in the case being considered, it remains to show that
\,${\vdash_{\GtwoLukA} \mathcal{H}_2}$.
For this, it is sufficient to establish that
\,${\vDash \widetilde{\mathcal{H}}_3}$ implies
\,${\vDash \widetilde{\mathcal{H}}_2}$.
Indeed, \,${\vdash_{\GthreeLukA} \widetilde{\mathcal{H}}_3}$\, and
the soundness of $\GthreeLukA$ (see Theorem \ref{GTh:CorrGthree}) 
guarantee that \,${\vDash \widetilde{\mathcal{H}}_3}$.
If we prove that the latter implies 
\,${\vDash \widetilde{\mathcal{H}}_2}$,\, 
then first we will obtain \,${\vdash_{\GtwoLukA} \widetilde{\mathcal{H}}_2}$\,
by the completeness of $\GtwoLukA$ for quantifier-free hypersequents
(see Proposition 14 in \cite{Ger2017}), 
and next we will get \,${\vdash_{\GtwoLukA} \mathcal{H}_2}$\, 
because a rule similar to the rule $\text{(ew)}^3$ in Lemma \ref{GLem:AtEwSplitAdmGthree}
is admissible for $\GtwoLukA$.

The hypersequent $\widetilde{\mathcal{H}}_2$ has the form:
$${\Rightarrow B(t_3,a_2,t_5) \,| \Rightarrow B(t_3,a_2,t_4) \,| \Rightarrow B(t_1,a_1,t_2)}\,;$$
and the hypersequent $\widetilde{\mathcal{H}}_3$ has the form:
\begin{gather*}
 \Rightarrow \SpV{q}_1 \,|\, \SpV{q}_1 \Rightarrow \SpV{q}_3 \,|\, \SpV{q}_3 \Rightarrow \SpV{q}_4
 \,|\, \SpV{q}_4 \Rightarrow \SpV{q}_5 \,|\, \SpV{q}_5 \Rightarrow B(t_3,a_2,t_5) \\
 |\, \SpV{q}_4 \Rightarrow B(t_3,a_2,t_4) \,|\, \SpV{q}_1 \Rightarrow \SpV{q}_2 \,|\, \SpV{q}_2 \Rightarrow B(t_1,a_1,t_2)\,.
\end{gather*}

For a hypersequent $\mathcal{H}$, we write \,${\nvDash \mathcal{H}}$
to denote that $\mathcal{H}$ is not valid.

The condition \,${\nvDash \widetilde{\mathcal{H}}_2}$ is equivalent to 
the existence of an interpretation $M_2$ and a valuation $\nu_2$ such that
these three inequalities hold:
$${1>|B(t_3,a_2,t_5)|_{M_2,\nu_2}}, \quad
{1>|B(t_3,a_2,t_4)|_{M_2,\nu_2}},   \quad
{1>|B(t_1,a_1,t_2)|_{M_2,\nu_2}}.$$

The condition \,${\nvDash \widetilde{\mathcal{H}}_3}$ is satisfied iff
there exist an interpretation $M_3$ and a valuation $\nu_3$ for which 
all these inequalities hold: 
\begin{gather*}
1 > |\SpV{q}_1|_{M_3,\nu_3} > |\SpV{q}_3|_{M_3,\nu_3} > |\SpV{q}_4|_{M_3,\nu_3} > |\SpV{q}_5|_{M_3,\nu_3} > |B(t_3,a_2,t_5)|_{M_3,\nu_3},\\
\phantom{xxxxxxxxxxxll} |\SpV{q}_4|_{M_3,\nu_3} > |B(t_3,a_2,t_4)|_{M_3,\nu_3},\\
|\SpV{q}_1|_{M_3,\nu_3} > |\SpV{q}_2|_{M_3,\nu_3} > |B(t_1,a_1,t_2)|_{M_3,\nu_3}. \phantom{xxxxxxxxxxxxxw}
\end{gather*}

Clearly, \,${\nvDash \widetilde{\mathcal{H}}_2}$ implies
\,${\nvDash \widetilde{\mathcal{H}}_3}$, as required in the given case.

In the general case, it is obvious that from $\widetilde{D}_3$ 
we can similarly construct a $\GtwoLukA$-proof search tree $\widetilde{D}_2$.
Then the assertion \,``${\nvDash \widetilde{\mathcal{H}}_2}$ implies 
\,${\nvDash \widetilde{\mathcal{H}}_3}$'' 
follows from the next observation, which is easily justified 
using induction on the height of the tree $\widetilde{D}_3$.

We can represent the hypersequent $\widetilde{\mathcal{H}}_3$ as
a directed acyclic graph by associating,
to each sequent member in $\widetilde{\mathcal{H}}_3$, a unique vertex and, 
to each sequent of the form ${F_1 \Rightarrow F_2}$,\, 
an edge from $F_1$ to $F_2$.
In this graph, there is exactly one source, and
all vertices corresponding to $\RPLA$-formulas are sinks.
The condition ${\nvDash \widetilde{\mathcal{H}}_3}$ is equivalent to 
the existence of an interpretation $M_3$ and a valuation $\nu_3$ such that, 
for each edge ${F_1 \Rightarrow F_2}$,\, 
the inequality ${|F_1|_{M_3,\nu_3} > |F_2|_{M_3,\nu_3}}$ holds and 
so does the inequality ${1 > |F|_{M_3,\nu_3}}$ for the source $F$ of the graph.
\end{proof}

\begin{theorem} 
\label{Gth:UndecGthreeProvProbl}
Let \,$\mathfrak{S}$ be a signature such that 
the validity problem for existential sentences of classical logic
over $\mathfrak{S}$ is undecidable.
Then the $\GthreeLukA$-provability problem for existential \,$\LukA$-sentences 
over $\mathfrak{S}$ is undecidable. 
\end{theorem}

\begin{proof}
By Theorem 21 in \cite{Ger2017}, 
the corresponding problem for $\GtwoLukA$ is undecidable; 
and the result follows by Theorem \ref{GTh:GoneGtwoGthreeEquivForPrenexFormulas}.
\end{proof}

\section{Conclusion}    

For the logics $\LukA$ and $\RPLA$, we presented the hypersequent 
calculus $\GthreeLukA$, whose rules are repetition-free and hp-invertible.

Theorem \ref{GTh:GOneProvGthreeProv} established above and 
Theorem 4 and Proposition 11 both given in \cite{Ger2017} 
ensure that any $\GLukA$-, $\GoneLukA$-, or $\GtwoLukA$-provable hypersequent
is provable in $\GthreeLukA$.
By Theorem \ref{GTh:GoneGtwoGthreeEquivForPrenexFormulas} in the present paper,
any prenex $\RPLA$-formula is provable or unprovable in 
$\GoneLukA$, $\GtwoLukA$, and $\GthreeLukA$ simultaneously.
From Theorem \ref{GTh:GoneGtwoGthreeEquivForPrenexFormulas} stated above and
Theorem 17 given in \cite{Ger2017}, 
it follows that any prenex $\LukA$-formula is $\GLukA$-provable 
iff it is $\GthreeLukA$-provable.

In essentially the same manner as in \cite[Section~4]{Ger2017},
we can formulate 
a free-va\-riable tableau modification $\TthreeLukA$ of the calculus $\GthreeLukA$
and describe a family of $\TthreeLukA$-proof search algorithms parameterized
by a fair tactic.
Then Theorem \ref{GTh:RearrangingGthreeProofAccordTactics}
(on constructing $\GthreeLukA$-proofs according to fair tactics)
will allow us to establish that any algorithm of the family 
constructs some $\TthreeLukA$-proof for any $\GthreeLukA$-provable sentence
(and so for any $\GLukA$-provable sentence).

Among problems for further research are the following.

1. Find out whether every $\LukA$-sentence (resp. $\RPLA$-sentence)
provable in $\GthreeLukA$ is provable in $\GLukA$ (resp. in $\GtwoLukA$). 

2. Investigate how complexity of formal proofs varies in passages from
one of the calculi mentioned to another.

3. Describe a nontrivial class $\mathbb{C}$ of hypersequent calculi 
in syntactic terms,
with every calculus of $\mathbb{C}$ having the proof-theoretic properties
established for $\GthreeLukA$. 
Cf., e.g., \cite{Nigam2016}, 
which gives sufficient conditions for several properties
of some sequent calculi, in particular, for invertibility of inference rules.

4. Develop a method for obtaining sound calculi of the class $\mathbb{C}$,
for first-order many-valued logics meeting some semantic conditions.
Cf. \cite{Bongini2016}, 
which solves a somewhat similar problem for a certain class of
propositional many-valued logics.

\newpage
\appendix
\section{\textnormal{Errata to the article \texorpdfstring{ 
  ``Infinite-valued first-order {\L}ukasiewicz logic: hypersequent calculi
  without structural rules and proof search for sentences in the prenex form''
  by A.~S.~Gerasimov, \emph{Siberian Advances in Mathematics}, Vol.~28, 
  No.~2 (2018), pp.~79--100 (\url{https://doi.org/10.3103/S1055134418020013})
 }{[7]}
}}
\label{secApp:Errata}

\newcommand{\ExtLukA}{\textnormal{F}\forall}

The above article is an English translation of the Russian article
published in \emph{Matematicheskie Trudy}, Vol.~20, No.~2 (2017), pp.~3--34
(\url{http://www.mathnet.ru/rus/mt321}).
Below the author corrects the most misleading inaccuracies 
introduced by a translator.

\begin{description}
\item[Page 79, line 14 from bottom]
``(3)~The Gentzen type sequent calculus $\ExtLukA$ for the logic of fuzzy inequalities''
\emph{should be}
``(3)~The Gentzen type sequent calculus for the logic of fuzzy inequalities $\ExtLukA$''.\footnote{
  In the (Russian) original: \selectlanguage{russian}
  ``(3)~Секвенциальное исчисление генценовского типа для логики нечетких неравенств 
  $\ExtLukA$ [1; 2], которая является расширением $\RPLA$.''
}

\item[Page 80, line 1 from bottom]
``the premise'' \emph{should be} ``each premise''.\footnote{
  In the original: \selectlanguage{russian}
  ``любая посылка правила вывода содержит заключение этого правила''.
}

\item[Page 82, line 10 from top]
``occurrence'' \emph{should be} ``repetition''.\footnote{
  In the original: \selectlanguage{russian}
  ``мультимножества $\Gamma$ и $\Delta$ повторяются в посылке.''
}

\item[Page 82, line 12 from top]
``false'' \emph{should be} ``unsound''.\footnote{
  In the original: \selectlanguage{russian}
  ``некорректного правила''.
}

\item[Page 83, lines 11--12 from bottom]
``We replace each proper parameter occurring in $\mathscr{G}$ 
and each proper semipropositional variable in $D$'' \emph{should be} 
``In $D$, we replace all proper parameters and 
proper semipropositional variables of $D$ occurring in $\mathscr{G}$''.\footnote{
  In the original: \selectlanguage{russian}
  ``$\GoneLukA$-вывод $D'$ получим, переименовав в $D$ 
  все входящие в $\mathscr{G}$ собственные параметры и
  собственные полупропозициональные переменные вывода $D$ 
  на новые попарно различные.''
}

\item[Page 84, line 17 from top]
``We claim that assertion (1) is equivalent to the following equivalent conditions:'' 
\emph{should be} ``Then assertion (1) is equivalent to the following:''.\footnote{
  In the original: \selectlanguage{russian}
  ``Тогда утверждение (1) эквивалентно следующему:''.
}

\item[Page 86, lines 8, 12, and 19 from top]
``sequence'' \emph{should be} ``sequent''.\footnote{
  In the original: \selectlanguage{russian}
  ``секвенция'' (в соответствующем падеже и числе).
}

\item[Page 87, lines 16 and 18 from top]
``consists of'' \emph{should be} ``contains''.\footnote{
  In the original: \selectlanguage{russian}
  ``содержит''.
}

\item[Page 89, line 13 from bottom]
``distinguished occurrences'' \emph{should be} ``the distinguished occurrence''.\footnote{
  In the original: \selectlanguage{russian}
  ``выделенного в посылке $(\star)$ 
  вхождения $\Rightarrow \forall x B$''.
}

\item[Page 90, line 2 from top]
``starting from the root of $D$'' \emph{should be} 
``in order of increasing their distances from the root of $D$''.\footnote{
  In the original: \selectlanguage{russian}
  ``в порядке удаления от корня дерева $D$''.
}

\item[Page 91, lines 23--24 from top]
``propositional logic'' \emph{should be} ``propositional classical logic''.\footnote{
  In the original: \selectlanguage{russian}
  ``пропозициональной классической логики''.
}

\item[Page 93, lines 3, 23, and 26 from top]
``table'' \emph{should be} ``tableau''.\footnote{
  In the original: \selectlanguage{russian}
  ``таблица'' (в соответствующем падеже). \selectlanguage{english}
  In all other places in the article, 
  this word is correctly translated by ``tableau''.
}

\item[Page 94, line 22 from top]
``system'' \emph{should be} ``system $\mathscr{S}_{\mathscr{H}_i}$''.\footnote{
  In the original: \selectlanguage{russian}
  ``систему $\mathscr{S}_{\mathscr{H}_i}$''.
}

\item[Page 95, line 20 from top]
\emph{Remove} ``otherwise,''.\footnote{
  In the original: \selectlanguage{russian}
  ``(3) закончить с ответом <<непревращаемы>>''.
}

\item[Page 95, line 31 from top]
``literals'' \emph{should be} ``distinct literals''.\footnote{
  In the original: \selectlanguage{russian}
  ``ровно 3 различных литерала''.
}

\item[Page 96, line 4 from top]
``expressed'' \emph{should be} ``bounded'' 
(\emph{or, to be closer to the original,}
``can be expressed by a polynomial'' \emph{should be} ``is polynomial'').\footnote{
  In the original: \selectlanguage{russian}
  ``выполняет полиномиальное от длины входа число операций''.
}

\item[Page 97, line 8 from bottom]
``The premise of each'' \emph{should be} ``Each premise of a''.\footnote{
  In the original: \selectlanguage{russian}
  ``каждая посылка контприменения''.
}

\item[Page 98, line 10 from top]
``Informally'' \emph{should be} ``Otherwise, informally''.\footnote{
  In the original: \selectlanguage{russian}
  ``Иначе мы, грубо говоря, вставим\ldots''
}

\item[Page 98, line 12 from top]
``distinguished occurrences'' \emph{should be} 
``the distinguished occurrence''.\footnote{
  In the original: \selectlanguage{russian}
  ``выделенного вхождения секвенции $\Rightarrow \forall x B$''.
}

\item[Page 99, line 11 from bottom]
\emph{Remove} ``, Vol. 2''.\footnote{
  In the original, no volume is mentioned intentionally
  (i.e., both volumes are referred to.)
}
\end{description}

\bigskip\bigskip
\end{document}